\documentclass[runningheads]{llncs}
\usepackage{graphicx}
\usepackage{graphics}
\usepackage{array}
\usepackage{amsmath}
\usepackage{amssymb}
\usepackage{comment}
\usepackage{setspace}
\usepackage{algorithm}
\usepackage{epstopdf}       
\usepackage{enumerate}     
\usepackage{xcolor}

\def \bb {\textcolor{black} }

\def\NC#1#2{$N^I(#1)\cup\{#1,#2\}$}
\def \K#1#2{$K_{#1,#2}$}
\def \NV#1{$\{#1\}\cup N^I(#1)$}
\def \nin {\notin}
\newtheorem{obs}{Observation}
\newtheorem{cl}{Claim}
\usepackage{fullpage}
\usepackage{algcompatible}   
 
\date{}
\title{The Hamiltonian Cycle in $K_{1,r}$-free Split Graphs - A Dichotomy}
\author{P.Renjith, N.Sadagopan} 
\institute{Indian Institute of Information Technology, Design and Manufacturing, Kurnool, India. \\Indian Institute of Information Technology, Design and Manufacturing, Kancheepuram, India. \\
\email{renjith@iiitk.ac.in, sadagopan@iiitdm.ac.in}}
\begin{document}
\maketitle \footnotetext{This work is partially supported by DST-ECRA Project - ECR/2017/001442.\\ Preliminary results on this problem appeared in CALDAM 2016.}
\vspace{-20pt}
 \begin{abstract}
In this paper, we investigate the well-studied Hamiltonian cycle problem (HCYCLE), and present an interesting dichotomy result on split graphs.  T. Akiyama et al. (1980) have shown that HCYCLE is NP-complete in planar bipartite graphs with maximum degree $3$.  Using this reduction, we show that HCYCLE is NP-complete in split graphs.  In particular, we show that the problem is NP-complete in $K_{1,5}$-free split graphs.  Further, we present polynomial-time algorithms for Hamiltonian cycle in $K_{1,3}$-free and $K_{1,4}$-free split graphs.  We believe that the structural results presented in this paper can be used to show similar dichotomy result for Hamiltonian path problem (HPATH) and other variants of HCYCLE.
 \end{abstract}
\section{Introduction}
The Hamiltonian cycle problem (HCYCLE) is a well-known decision problem which asks for the presence of a spanning cycle in a graph.
Hamiltonian problems play a significant role in various research areas such as operational research, physics and genetic studies \cite{app1,app2,app3}.  
This well-known problem has been studied extensively in the literature, and is NP-complete in general graphs.  
Various sufficient conditions for the existence of a Hamiltonian cycle \cite{west} were introduced by G.A. Dirac, O. Ore, J.A. Bondy, and A. Ainouche, which were further  generalized by A. Kemnitz \cite{KemnSchi}. 
A general study on the sufficient conditions were produced by H.J. Broersma \cite{s1} and \bb{separately by} R.J. Gould \cite{s2,s3}.
Hamiltonicity has been looked at with respect to various structural parameters, the popular one \bb{being} toughness.
A relation between graph toughness, introduced by V. Chv\'{a}tal \cite{chvataltough} and Hamiltonicity has been well studied.  A detailed survey on Hamiltonicity and toughness is presented in \cite{s2,bauer}.
In split graphs, it \bb{was} proved by D. Kratsch, J. Lehel, and H. Muller \cite{kratsch} that $\frac{3}{2}$ tough split graphs are Hamiltonian and due to Chv\'{a}tal's result \cite{chvataltough}, Hamiltonian graphs are $1$-tough.  Therefore, the split graphs are Hamiltonian if and only if the toughness is in the range $[\frac{3}{2},1]$.  \\

\noindent
On \bb{the} algorithmic front, HCYCLE is NP-complete in chordal \cite{bertossi}, chordal bipartite \cite{muller}, planar \cite{tarjanplanar}, and bipartite \cite{akiyama} graphs.  Further, T. Akiyama \cite{akiyama} has shown that the problem is NP-complete in bipartite graphs with maximum degree $3$.  There is a simple reduction for HCYCLE in bipartite graphs of maximum degree $3$ to HCYCLE in split graphs which we show as part of our dichotomy.
In spite of the hardness of HCYCLE in various graph classes, nice polynomial-time algorithms have been obtained in interval, circular arc, $2$-trees, and distance hereditary graphs \cite{keil,hungcir2,shih,hungdis}.
\\\\
Split graphs are  a popular subclass of chordal graphs and on which R.E. Burkard and P.L. Hammer\cite{burkard} presented a necessary condition \bb{for the presence of Hamiltonian cycle}.  Subsequently, N.D. Tan and L.X. Hung \cite{tan} have shown that the necessary condition of \cite{burkard} is sufficient for some special split graphs.  
In this paper, we shall revisit Hamiltonicity restricted to split graphs and present a dichotomy result.  We show that HCYCLE is NP-complete in $K_{1,5}$-free split graphs and polynomial-time solvable in $K_{1,4}$-free split graphs.  It is important to note that very few NP-complete problems have dichotomy results in the literature \cite{dicho1,dicho2,dicho3}.
\\\\
We use standard basic graph-theoretic notations.  Further, we follow \cite{west}.  All the graphs we mention are simple, and unweighted.  Graph $G$ has vertex set $V(G)$ and edge set $E(G)$ which we denote using $V,E$, respectively.  For \emph{minimal vertex separator}, \emph{maximal clique}, and \emph{maximum clique} we use the standard definitions.  A graph $G$ is 2-connected if every minimal vertex separator is of size at least two.  Split graphs are $C_4,C_5,2K_2$-free graphs and the vertex set of a split graph can be partitioned into a clique $K$ and an independent set $I$.  For a split graph with vertex partition $K$ and $I$, we assume $K$ to be a maximum clique.  For every $v\in K$ we define $N^I(v)=N(v)\cap I$, where $N(v)$ denotes the neighborhood of vertex $v$.  $d^I(v)=|N^I(v)|$ and $\Delta^I=\max\{d^I(v) : v\in K\}$.  \bb{For a graph $G$ and $S\subseteq V(G)$, $N(S)=\bigcup\limits_{v\in S}N(v)$.}   
For a cycle or a path $C$, we use $\overrightarrow{C}$ to represent an ordering of the vertices of $C$ in one direction (forward direction) and $\overleftarrow{C}$ to represent the ordering in the other direction.  $u\overrightarrow{C}\bb{v}$ represents the \bb{ordering of} vertices from $u$ to $v$ in $C$.  For two paths $P$ and $Q$, $P\cap Q$ denotes $V(P)\cap V(Q)$.  For simplicity, we use $P$ to denote the underlying set $V(P)$. \bb{For $S\subseteq V(G)$, c(G-S) denotes the number of connected components of the graph G-S.}
\vspace*{-12pt}
\section{HCYCLE in split graphs - polynomial results}
\vspace{-10pt}
In this section we present structural results on some special split graphs using which we can find Hamiltonian cycle in such graphs.  
In particular, we explore the structure of $K_{1,3}$-free split graphs and $K_{1,4}$-free split graphs.
\begin{lemma} \cite{dicho1} \label{lemvileq3}
For a claw-free split graph $G$,  if $\Delta^I=2$, then $|I|\leq 3$. 
\end{lemma} 
\vspace*{-10pt}
\begin{obs}\label{del1hamil}
If $G$ is a $2$-connected split graph with \bb{$\Delta^I\le1$}, then $G$ has a Hamiltonian cycle. 
\end{obs}
\vspace*{-10pt}
\begin{lemma} \label{k13hamil}
Let $G$ be a $K_{1,3}$-free split graph.  Then $G$ contains a Hamiltonian cycle if and only if $G$ is $2$-connected.
\end{lemma}
\begin{proof}
Necessity is trivial.  For the sufficiency, we consider the following cases.\\
\emph{Case} 1: $|I|\geq 4$. 
As per Lemma \ref{lemvileq3}, $\Delta^I\le1$ and by Observation \ref{del1hamil}, $G$ has a Hamiltonian cycle.\\
\emph{Case} 2:  $|I|\le3$. If $\Delta^I=1$, then by Observation \ref{del1hamil}, $G$ has a Hamiltonian cycle.  When $\Delta^I>1$, note that either $|I|=2$ or $|I|=3$. \\
\textbf{(a)} $|I|=2$, i.e., $I=\{s,t\}$.  Since $\Delta^{I}=2$, there exists a vertex $v\in K$ such that $d^{I}(v)=2$ and $N^{I}(v)=\{s,t\}$.  
Let $S=N(s)$ and $T=N(t)$.  \bb{Suppose $K\neq S\cup T$, then there exists $z\in K$ such that $N^I(v)\cap N^I(z)=\emptyset$. Further, $\{z,v,s,t\}$ induces a claw, a contradiction. Therefore, $K=S\cup T$.}  Suppose the set $S\setminus T$ is empty, then the vertices $T\cup \{t\}$ induces a clique, larger in size than $K$, contradicting the maximality of $K$.  Therefore, the set $S\setminus T$ is non-empty.  Similarly, $T\setminus S\neq \emptyset$.  
It follows that, $|K|\geq 3$; let $x,w\in K$ such that $x\in T\setminus S$ and $w\in S\setminus T$.   
Further, $(w, s, v, t, x, v_1, v_2, \ldots, v_k,w)$ is a Hamiltonian cycle in $G$ where $\{v_1, v_2, \ldots, v_k\}=K\setminus \{v,w,x\}$.  \\
\textbf{(b)} If $|I|=3$ and let $I=\{s,t,u\}$.  
Since $G$ is a $2$-connected claw free graph, clearly $|K|\ge3$.  
Since $\Delta^{I}=2$, there exists $v\in K$ such that \bb{$v$ is adjacent to two vertices in $I$. Without loss of generality, let} $N^{I}(v)=\{s,t\}$.  Further, since $G$ is claw free, for every vertex $w\in K$, $N^{I}(w)\cap \{s,t\}\neq \emptyset$.  Let $S=N(s)$ and $T=N(t)$.  Suppose the set $S\setminus T$ is empty, then the vertices $T\cup \{t\}$ induces a clique, larger in size than $K$, contradicting the maximality of $K$.  Therefore, the set $S\setminus T\neq \emptyset$.  Similarly, $T\setminus S\neq \emptyset$.  Note that $|N(u)|\ge 2$ as $G$ is $2$-connected.  If $u$ is adjacent to a vertex $v\in S\cap T$, then $\{v\}\cup N^I(v)$ induces a claw.  Therefore, $N(u)\cap (S\cap T)=\emptyset$.  It follows that $u$ is adjacent to some vertices in $S\setminus T$ or $T\setminus S$.  Suppose $N(u)\cap (S\setminus T)=\emptyset$ and $N(u)\cap (T\setminus S)\neq \emptyset$, then there exists a vertex $x\in T\setminus S$ such that $N^I(x)=\{t,u\}$.  Since $S\setminus T\neq\emptyset$, there exists $w\in S\setminus T$ and $\{x,u,t,w\}$ induces a claw, a contradiction.  Therefore, $N(u)\cap (S\setminus T)\neq \emptyset$ and similarly, $N(u)\cap (T\setminus S)\neq \emptyset$.  \bb{It follows that there exists a vertex $x\in (T\setminus S)$, $w\in (S\setminus T)$ such that $w,x\in N(u)$.  If $|K|=3$, then $(w,s,v,t,x,u,w)$ is a Hamiltonian cycle in $G$.  Further, if $|K|> 3$, then note that there exists $y\in K$ such that $y\notin\{w,v,x\}$ is adjacent to exactly two vertices of $I$. Without loss of generality, let $ys,yt\in E(G)$.  $(y,s,v,t,x,u,w,w_1,\ldots,w_l)$ is the desired Hamiltonian cycle where $\{w_1,\ldots, w_l\}=K\setminus \{w,v,x,y\}$.}  This completes the case analysis, and the proof of Lemma \ref{k13hamil}.  $\hfill \qed$
\end{proof}
%
\textbf{$\bf K_{1,4}$-free split graphs}\\ 
Now we shall present structural observations in $K_{1,4}$-free split graphs.  The structural results in turn give a polynomial-time algorithm for the Hamiltonian cycle problem, which is one part of our dichotomy.
From Observation \ref{del1hamil}, it follows that a split graph $G$ with $\Delta^I=1$ has a Hamiltonian cycle if and only if $G$ is $2$-connected.  It is easy to see that for $K_{1,4}$-free split graphs, $\Delta^I\le3$, and thus we analyze such a graph $G$ in two variants, \bb{$\Delta^I_{G}\le2$, and $\Delta^I_G=3$}.  


\noindent
The next theorem shows a necessary and sufficient condition for the existence of Hamiltonian cycle in split graphs with $\Delta^I=2$.  
We define the notion short cycle in a $K_{1,4}$-free split graph $G$.  Consider the subgraph $H$ of $G$ where $V_a=\{u\in I : d(u)=2\}$, $V_b=N(V_a)$, $V(H)=V_a\cup V_b$ and \bb{$E(H)=\{uv \in E(G): u\in V_a, v\in V_b\}$}.  Clearly, $H$ is a bipartite subgraph of $G$.  Let $C$ be an induced cycle in $H$ such that $V(K)\setminus V(C)\ne\emptyset$.  We refer to $C$ in $H$ as a \emph{short cycle} in $G$, \bb{if such a cycle exists}. 
\begin{theorem}\label{deltale2}
Let $G$ be a $2$-connected, $K_{1,4}$-free split graph with $\Delta^I\le2$.  Then $G$ has a Hamiltonian cycle if and only if there are no short cycles in $G$.  
\end{theorem}
\begin{proof}
Necessity is trivial as, if there exists a short cycle $D$, then $c(G-S)>|S|$ where $S=V(D)\cap K$.  For sufficiency: 
if $|I|=|K|$ and $G$ has no short cycles, then clearly, $H$ is a spanning cycle of $G$.  Further, if $G$ has no short cycles with $|K|>|I|$, \bb{and $H$ is non-empty}, then note that $H$ is a collection of paths $P_1,\ldots,P_i$, all of them are having end vertices in $K$.  Let $V'=I\setminus V_a$.  If $V'=\emptyset$, then it is easy to join the paths using clique edges to get a Hamiltonian cycle of $G$.  Otherwise we partition the vertices in $V'$ into three sets $V_2,V_1,V_0$ where \bb{$V_2=\{u\in V' : N(u)\cap P_j\neq \emptyset$ and $N(u)\cap P_{k}\neq \emptyset, 1\le j,k\le i \text{ and } j\neq k\}$},  $V_1=\{u\in V' : N(u)\cap P_j\neq \emptyset, 1\le j\le i$ and $u\notin V_2\}$,  $V_0=\{u\in V' : N(u)\cap P_j=\emptyset, 1\le j\le i\}$.  From the definitions, vertices in $V_2$ are adjacent to the end vertices of at least two paths, vertices in $V_1$ are adjacent to the end vertices of exactly one path and that of $V_0$ are not adjacent to the end vertices of any paths.  Now we obtain two graphs $H_1$, and $H_2$ from $H$ and finally, we see that $H_2$ is a collection of paths containing all the vertices of $I$, all  those paths having end vertices in $K$.  
We iteratively add the vertices in $V_2$ and $V_1$ into $H$ to obtain $H_1$, based on certain preferences till $V_2=V_1=\emptyset$.  
We pick a vertex ( from $V_2$ if $V_2\ne\emptyset$, otherwise from $V_1$ ) and add it to $H$.  If we add a vertex $u$ from $V_2$, then we join two arbitrary paths each has its end vertex adjacent to $u$.  Therefore, the addition of a vertex from $V_2$ reduces the number of paths in $H$ by one.  If $u\in V_1$, then $u$ is added to $H$ in such a way that one of its end vertex is an end vertex of a path and the other is not an end vertex of any paths.  Clearly, such two vertices are possible due to the fact that $d(u)\ge3$.  Note that in this case, one of the paths in $H$ gets its size increased, still in both cases all the paths have their end vertices in $K$.  We add the vertices in such a way that the vertices in $V_2$ gets preference over that of $V_1$.  Also, after every addition of a vertex, we add the new vertex to $V'$ and re-compute the partitions $V_2,V_1$ and $V_0$.  Observe that once the set $V_2$ is empty, and a vertex from $V_1$ is added, the re-computation may result in a case where $V_2\neq \emptyset$.  
Once $V_2=V_1=\emptyset$, the graph $H_1$ obtained is a collection of one or more paths having end vertices in $K$. \\\\
Now we continue the addition by including the vertices of $V_0$.  A vertex $u\in V_0$ is added in such a way that it forms a $P_3$ with two of its arbitrary neighbors in $K$.  Evidently, the addition of vertices from $V_0$ increases the number of paths by one.  Note that the addition of vertices from $V_0$ may result in $V_2\ne\emptyset$ or $V_1\ne\emptyset$.  Therefore, when we iteratively add the vertices to $H_1$, we give first preference to the vertices in $V_2$, then to the vertices in $V_1$ (if $V_2=\emptyset$) and finally to that in $V_0$ (if $V_2=V_1=\emptyset$).  Here also, after each addition of a vertex $u$, we add $u$ to $V'$ and re-compute the partitions.  $H_2$ is the graph obtained by iteratively adding all the vertices left in $I$.  Finally, in $H_2$ we have a collection of one or more paths with end vertices in $K$.  It is easy to see that the paths could be joined using clique edges to get a Hamiltonian cycle in $G$.  This completes the proof. $\hfill\qed$ 
\end{proof}
Having presented a characterization for the Hamiltonian cycle problem in split graphs with $\Delta^I\le2$, we shall now present our main result, which is a necessary and sufficient condition for the existence of Hamiltonian cycle in $K_{1,4}$-free split graphs.\\

\noindent
\textbf{Main result: } {\it 
Let $G$ be a $2$-connected, $K_{1,4}$-free split graph with $|K|\ge|I|\ge8$. 
$G$ has a Hamiltonian cycle if and only if there are no induced short cycles in $G$.  
Further, finding such a cycle is polynomial-time solvable.\\
}

\noindent 
To prove our main result, we make use of the following claims, which we shall discuss next in detail.

\newpage
\noindent 
\textbf{Claim A}: {\em For a split graph $G$ with $\Delta^I=3$, let $v\in K, d^I(v)=3$, and $U=N^I(v)$.  If $G$ is $K_{1,4}$-free, then $N(U)=K$.  }
\begin{proof}
If not, let $w\in K$ such that $w\notin N(U)$.  Clearly, $N^I(v)\cup \{w,v\}$ induces a $K_{1,4}$, a contradiction. $\hfill \qed$
\end{proof}
Note that for a $K_{1,4}$-free split graph $G$, $\Delta^I\le3$ and thus the left over case to analyze is when $\Delta^I=3$.  
When $\Delta^I=3$, there exists a vertex \bb{$v\in K$} with $d^I(v)=3$.  \bb{We obtain $G'=G-N^I(v)$, and the following claim is observable on $G'$.}\\ 

\noindent
\textbf{Claim B}:
{\em For a $K_{1,4}$-free split graph $G$ with $\Delta^I=3$, let $v\in K$ such that $d^I(v)=3$, and the split graph $G'=G-N^I(v)$.  Then, $\Delta^I_{G'}\le 2$.  }
\begin{proof}
Otherwise, if there exists $x\in K, d^I_{G'}(x)=3$, then $N^I(x)\cup \{x,v\}$ induces a $K_{1,4}$ in $G$, a contradiction.$\hfill \qed$
\end{proof}

\noindent
Let $G$ be a 2-connected $K_{1,4}$-free split graph with $\Delta^I=3$, $|K|\ge|I|\ge8$.  If there are no induced short cycles in $G$, then we consider the subgraph $G'=G-N^I(v)$ where $v\in K,N^I(v)=\{v_1,v_2,v_3\}$. By Claim B, note that $\Delta^I_{G'}\le 2$ and by Theorem \ref{deltale2}, there exists a Hamiltonian cycle in $G'$. By the constructive proof of Theorem \ref{deltale2}, in $H$ there exists a collection of  vertex disjoint paths $\mathbb{C}$.  Note that each path in $\mathbb{C}$ alternates between an element in $K$ and an element in $I$, and all the paths are \bb{have their} end vertices in $K$.  Therefore the paths \bb{have an} odd number of vertices.  Thus, $\mathbb{C}=\mathbb{P}_1\cup \mathbb{P}_3,\ldots,\mathbb{P}_{2i+1}$, where $\mathbb{P}_j$ is the set of maximal paths of size $j$ where for every $P\in\mathbb{P}_j$, there does not exists $P'\in\mathbb{C}$ such that $E(P)\subset E(P')$.  
A path $P_a\in\mathbb{C}$ is defined on the vertex set $V(P_a)=\{w_1,\ldots,w_j,x_1,\ldots,x_{j-1}\}$,  $E(P_a)=\{w_ix_i : 1\leq i\leq j-1\}\cup \{w_kx_{k-1} : 2\leq k\leq j\}$ such that $\{w_1,\ldots,w_j\}\subseteq K$, $\{x_1,\ldots,x_{j-1}\}\subseteq I$.   We denote such a path as $P_a=P(w_1,\ldots,w_j;x_1,\ldots,x_{j-1})$ 
We shall now present our structural observations on paths in $\mathbb{C}$.
\begin{cl}\label{xunivge11}
 If there exists a path $P_a\in \mathbb{P}_i,i\ge11$ such that $P_a=P(w_1,\ldots,w_j;x_1,\ldots,x_{j-1}), j\ge6$, then there exists $v_1\in N^I(v)$ such that \bb{$v_1w_l\in E(G), {2\le l\le j-1}$}.   
 \end{cl}
\begin{proof}
First we show that for every non-consecutive $2\le i,k\le j-1$, for the pair of vertices $w_i,w_k$, $v_1w_i\in E(G)$ and $v_1w_k\in E(G)$.
Suppose, if exactly one of $w_i,w_k$ is adjacent to $v_1$, say $w_iv_1\in E(G)$, then \NC{w_i}{w_k} has an induced \K14.  If $v_1w_i\nin E(G)$ and $v_1w_k\nin E(G)$, then by Claim A, there exists an adjacency for $w_i,w_k$ in $v_2,v_3$.  Further, if either $v_2$ or $v_3$ is adjacent to both $w_i,w_k$, then $v_1=v_2$, or $v_1=v_3$, and the claim is true.  Therefore, we shall assume without loss of generality, $v_2w_i\in E(G)$ and $v_2w_k\nin E(G)$.  This implies that \NC{w_i}{w_k} has an induced \K14, a contradiction.  Since the above observation is true for all such pair of vertices in \bb{$W=\{w_l : 2\le l\le j-1\}$}, and $|W|\ge4$, it follows that there exists $v_1\in N^I(v)$ such that $v_1w_l\in E(G), {2\le l\le j-1}$.   
$\hfill\qed$
\end{proof}
\begin{cl}\label{no13}
$\mathbb{P}_i=\emptyset,{i\ge13}$.
\end{cl} 
\begin{proof} 
Assume for a contradiction that there exists $P_{a}\in \mathbb{P}_i,{i\ge13}$.  Let $P_a=P(w_1,\ldots,w_j;x_1,\ldots,x_{j-1})$, ${j\ge7}$.  From Claim \ref{xunivge11} there exists $v_1\in N^I_G(v)$ such that $v_1w_k\in E(G),{2\le k\le j-1}$.  
Since the clique is maximum in $G$, there exists $s\in K$ such that $v_1s\notin E(G)$.  Further, there exists at least three vertices in $x_1,\ldots,x_{j-1}$ adjacent to $s$, otherwise, for some $2\le r\le j-1$, $N^I_G(w_r)\cup \{w_r,s\}$ induces a $K_{1,4}$.  Finally, from Claim A, either $v_2s\in E(G)$ or $v_3s\in E(G)$.  It follows that $\{s\}\cup N^I_G(s)$ has an induced $K_{1,4}$, a contradiction.  $\hfill \qed$ 
\end{proof}
\begin{cl}\label{xunivtwopaths}
Let $P_a=P(w_1,\ldots,w_i;x_1,\ldots,x_{i-1}), {i\ge3}$, and $P_b=P(s_1,\ldots,s_j;t_1,\ldots,t_{j-1}), {j\ge3}$ be arbitrary paths in $\mathbb{C}$.  Then there exists $v_1\in N^I(v)$ such that $\forall~{2\le l\le i-1}, v_1w_l\in E(G)$, and $\forall~{2\le m\le j-1},v_1s_m\in E(G)$. 
\end{cl} 
\begin{proof}
We shall consider every pair of vertices $w_l,s_m$ and show that $v_1w_l,v_1s_m\in E(G)$.
Suppose, if exactly one of $w_l,s_m$ is adjacent to $v_1$, say $w_lv_1\in E(G)$, then \NC{w_l}{s_m} has an induced \K14.  If $v_1w_l\nin E(G)$ and $v_1s_m\nin E(G)$, then by Claim A, there exists an adjacency for $w_l,s_m$ in $v_2,v_3$.  Further, if either $v_2$ or $v_3$ is adjacent to both $w_l,s_m$, then $v_1=v_2$, or $v_1=v_3$, and the claim is true.  Therefore, we shall assume without loss of generality, $v_2w_l\in E(G)$ and $v_2s_m\nin E(G)$.  This implies that \NC{w_l}{s_m} has an induced \K14, a contradiction.  $\hfill\qed$
\end{proof}
\begin{corollary} \label{xunivthreepaths}
Let $P_a=P(w_1,\ldots,w_i;x_1,\ldots,x_{i-1}), {i\ge3}$,  $P_b=P(s_1,\ldots,s_j;t_1,\ldots,t_{j-1}), {j\ge3}$ and \\ $P_c=P(y_1,\ldots,y_k;z_1,\ldots,z_{k-1}), {k\ge3}$ be arbitrary paths in $\mathbb{C}$.  Then there exists $v_1\in N^I(v)$ such that \\ $\forall~{2\le l\le i-1}, v_1w_l\in E(G)$, $\forall~{2\le m\le j-1},v_1s_m\in E(G)$, and $\forall~{2\le n\le k-1}, v_1y_n\in E(G)$. 
\end{corollary}
\begin{proof}
From Claim \ref{xunivtwopaths}, $\forall~{2\le l\le i-1}, v_1w_l\in E(G)$, and $\forall~{2\le m\le j-1},v_1s_m\in E(G)$.  Similarly,\\ $\forall~{2\le l\le i-1}, v_1w_l\in E(G)$, and $\forall~{2\le n\le k-1},v_1y_n\in E(G)$.  Thus the corollary follows from Claim \ref{xunivtwopaths}. $\hfill\qed$
\end{proof}
\begin{cl}\label{no115} 
 If there exists $P_a\in \mathbb{P}_{11}$, then $\mathbb{P}_j=\emptyset,j\ne11,j\ge5$.
\end{cl}
\begin{proof}
Assume for a contradiction that there exists such a path $P_b\in \mathbb{P}_j, j\ge5$.  Let $P_a=(w_1,\ldots,w_6;x_1,\ldots,x_5)$ and $P_b=(s_1,\ldots,s_r;t_1,\ldots,t_{r-1})$, $r\ge3$.  From Claim \ref{xunivge11}, there exists a vertex $v_1\in N^I(v)$, such that $v_1w_i\in E(G), 2\le i\le5$ and from Claim \ref{xunivtwopaths}, $v_1s_j\in E(G),2\le j\le{r-1}$.  
Now we claim $v_1w_1\in E(G)$.  Otherwise, by Claim A, $v_2w_1$ or $v_3w_1$ is in $E(G)$.  Observe that either $w_1x_2\in E(G)$ or $w_1x_3\in E(G)$, otherwise $N^I(w_3)\cup \{w_3,w_1\}$ induces a $K_{1,4}$.  Similarly, either $w_1x_4\in E(G)$ or $w_1x_5\in E(G)$.  Now $\{w_1\}\cup N^I(w_1)$ induces a $K_{1,4}$.  Using similar argument, we establish $v_1w_6\in E(G)$.  Since the clique is maximal, there exists a vertex $w'\in K$ such that $v_1w'\notin E(G)$.  We see the following cases.  \\{\bf {\em Case} 1:} $w'=s_1$. By Claim A, $s_1v_2\in E(G)$ or $s_1v_3\in E(G)$.  Further, $s_1x_2\in E(G)$, otherwise $N^I(w_2)\cup \{w_2,s_1\}$ induces a $K_{1,4}$ or $N^I(w_3)\cup \{w_3,s_1\}$ induces a $K_{1,4}$.  Similarly, $s_1x_4\in E(G)$.  Now $\{s_1\}\cup N^I(s_1)$ induces a $K_{1,4}$.  Similarly, we could establish a contradiction if $w'=s_r$. \\ {\bf {\em Case}  2:} $w'\notin P_b$.  By Claim A, $w'v_2\in E(G)$ or $w'v_3\in E(G)$.  Also due to the similar reasoning for $s_1$, $w'x_2,w'x_4\in E(G)$ .  Now, either $t_1w'\in E(G)$ or $t_2w'\in E(G)$, otherwise $N^I(s_2)\cup \{s_2,w'\}$ induces a $K_{1,4}$.  Finally, $\{w'\}\cup N^I(w')$ induces a $K_{1,4}$, a contradiction.  Therefore, $P_b$ does not exist.  This completes the case analysis and the proof.  $\hfill\qed$
\end{proof}
%
\begin{cl}\label{HC11} 
 If there exists $P_a\in\mathbb{P}_{11}$, then $G$ has a Hamiltonian cycle.
\end{cl}
\begin{proof}
Let $P_a=(w_1,\ldots,w_6;x_1,\ldots,x_5)$.  From Claim \ref{xunivge11}, there exists a vertex say $v_1\in N^I_G(v)$, such that $v_1w_i\in E(G)$, $2\le i\le 5$.  From the proof of the previous claim, $v_1w_1,v_1w_6\in E(G)$.  Since the clique is maximal, there exists $w'\in K$, such that $w'v_1\notin E(G)$.  By Claim A, $w'v_2\in E(G)$ or $w'v_3\in E(G)$.  Without loss of generality, let $w'v_2\in E(G)$.   We claim $w'x_2\in E(G)$ and $w'x_4\in E(G)$, otherwise for some $2\le i\le 5$, $N^I(w_i)\cup \{w_i,w'\}$ induces a $K_{1,4}$.  One among $v_2,x_2,x_4$ is adjacent to $w_1$, otherwise $N^I(w')\cup \{w_1,w'\}$ induces a $K_{1,4}$.  Similar argument holds good with respect to the vertex $w_6$.  \bb{Note that every $t\in\{v,w',w_1,\ldots,w_6\}$ has three neighbors, none of which is $v_3$.  And since $d^I(t)\le3$, $v_3$ cannot be adjacent to any $t\in\{v,w',w_1,\ldots,w_6\}$ But $G$ is 2-connected and therefore $v_3$ has at least 2 neighbors in $K$.} 
\bb{Let $w''\in K$, $w''\notin \{v,w',w_1,\ldots,w_6\}$ and $w''v_3\in E(G)$.}  Now we claim $w''v_1\in E(G)$.  If not, for some $1\le j\le 6$, $N^I(w_j)\cup \{w_j,w''\}$ induces a $K_{1,4}$.  Finally $(w_1\overrightarrow{P_a}w_{6},v_1,w'',v_3,v,v_2,w')$ is a $(w_1,w')$ path containing all the vertices of $P_a\cup \{v,w',w''\}\cup N^I(v)$, which could be easily extended to a Hamiltonian cycle in $G$ using clique edges to join other vertex disjoint paths. $\hfill\qed$
\end{proof}
\begin{cl}\label{no95} 
 If there exists $P_a\in\mathbb{P}_9$, then $\mathbb{P}_j=\emptyset,j\ne9,j\ge5$ \bb{and $|\mathbb{P}_9|\le 1$}.
\end{cl}
\begin{proof}
Assume for a contradiction that there exists such a path $P_b\in \mathbb{P}_j,j\ge5$.  Let $P_a=(w_1,\ldots,w_5;$ $x_1,\ldots,x_4)$ and $P_b=(s_1,\ldots,s_r;t_1,\ldots,t_{r-1})$, $r\ge3$.  From Claim \ref{xunivtwopaths}, there exists $v_1\in N^I_G(v)$ such that $v_1w_2,v_1w_3,v_1w_4,$ $v_1s_i\in E(G)$, $2\le i \le r-1$.  Now we claim that $w_1v_1\in E(G)$.  Suppose \bb{not. Then,} by Claim A either $w_1v_2\in E(G)$ or $w_1v_3\in E(G)$. 
Observe that $w_1x_3\in E(G)$, otherwise $N^I(w_3)\cup \{w_3,w_1\}$ or $N^I(w_4)\cup \{w_4,w_1\}$ induces a $K_{1,4}$.  Similarly, either $w_1t_1\in E(G)$ or $w_1t_2\in E(G)$  otherwise $N^I(s_2)\cup \{s_2,w_1\}$ induces a $K_{1,4}$.  It follows that $\{w_1\}\cup N^I(w_1)$ induces a $K_{1,4}$.  This contradicts the assumption that $w_1v_1\nin E(G)$, and thus $w_1v_1\in E(G)$.  Similar argument holds good with other end vertices of paths $P_a,P_b$, and hence, $w_5v_1,s_1v_1,s_rv_1\in E(G)$.  
Finally, we claim that $w'v_1\in E(G)$ for every vertex $w'\in K$, where $w'\notin\{w_1,\ldots,w_5,s_1,\ldots,s_r\}$.  If not, let $w'v_1\notin E(G)$.  By Claim A, either $v_2w'$ or $v_3w'$ is in $E(G)$.  Further, there exists at least two vertices in $x_1,\ldots,x_5$ adjacent to $w'$, otherwise, for some $2\le i\le4$, $N^I_G(w_i)\cup \{w_i,w'\}$ induces a $K_{1,4}$.  Now either $w't_1\in E(G)$ or $w't_2\in E(G)$, if not, $N^I(s_2)\cup \{s_2,w'\}$ induces a $K_{1,4}$.  Therefore, $\{w'\}\cup N^I(w')$ induces a $K_{1,4}$, a contradiction. 
From the above, we conclude that \bb{$\{v_1\}\cup K$} is a larger clique, which finally contradicts the existence of $P_b$.  Thus if $P_a\in \mathbb{C}$, then there does not exists such a path $P_b$.  This completes the proof. $\hfill\qed$  
\end{proof}
\noindent
In the following claims to show the existence of Hamiltonian cycle, we \bb{present} a  constructive approach in which we produces a $(x,y)$-path where $x,y\in K$.  The path is obtained by joining some paths in $\mathbb{C}$ using the vertices in $N^I(v)$.  Therefore such a \emph{desired path} is sufficient to show that $G$ has a Hamiltonian cycle, which is in turn obtained by joining all such vertex disjoint paths using clique edges.
%
\begin{cl}\label{HC931}
If there exists $P_a\in\mathbb{P}_9$ and $G$ has no short cycles, then $G$ has a Hamiltonian cycle.
\end{cl}
\begin{proof}
 Let $P_a=(w_1,\ldots,w_5;x_1,\ldots,x_4)$.  Since $|I|\geq 8$ and by Claim \ref{no95}, there exists at least one more path $P_b\in\mathbb{P}_3$ such that $P_b=(s_1,s_2;t_1)$ 
There exists a vertex in $N^I(v)=\{v_1,v_2,v_3\}$ adjacent to $w_2$, say $v_1w_2\in E(G)$.  Note that $v_1w_4\in E(G)$, otherwise $\{w_2,w_4\}\cup N^I(w_2)$ \bb{has an induced} $K_{1,4}$.  Since the clique is maximal, there exists a non-adjacency for $v_1$ in $K$, and based on \bb{this} non-adjacency, we see the following cases.  
\begin{enumerate}[  ]
\item \emph{\bf {\em Case} 1:}  $v_1w_1\nin E(G)$ or $v_1w_5\nin E(G)$.  Without loss of generality, we shall assume $v_1w_1\notin E(G)$.  Note that one of $v_2,v_3$ is adjacent to  $w_1$, say $v_2w_1\in E(G)$.  
Note that $w_1x_3\in E(G)$ or $w_1x_4\in E(G)$, if not \NC{w_4}{w_1} has an induced \K14.  Note that $w_3v_1\in E(G)$ or $w_3v_2\in E(G)$ or $w_3v_3\in E(G)$. \bb{If $w_1x_4, w_3v_1\in E(G)$, then \NC{w_3}{w_1} has an induced \K14. Similarly, if $w_1x_4, w_3v_3\in E(G)$, then \NC{w_3}{w_1} has an induced \K14.}
Therefore, if $w_1x_4\in E(G)$, then $w_3v_1\nin E(G)$ and $w_3v_3\nin E(G)$.  There exists four possibilities as follows.
\begin{enumerate}[  ]
\item \emph{\bf {\em Case} 1.1:} $w_1x_3,w_3v_1\in E(G)$.  We shall see the adjacency of the vertices $s_1$ and $s_2$ with respect to $\{v_1,v_2,v_3\}$.  We observe that either $v_1s_i,v_2s_i\in E(G)$ or  $v_1s_i,x_1s_i\in E(G)$ or $v_1s_i,x_3s_i\in E(G)$, $i\in \{1,2\}$, \bb{otherwise for some $z\in \{w_1,\ldots, w_4,v\}$, $\{z,s_i\}\cup N^I(z)$ has an induced $K_{1,4}$}.  Clearly, $v_1s_i\in E(G)$.  If $w_5v_1\in E(G)$, then $w_5v_2$ or $w_5x_1$ or $w_5x_3$ is in $E(G)$, otherwise $N^I(w_1)\cup \{w_1,w_5\}$ induces a $K_{1,4}$.  Further, note that $d^I(w_j)=d^I(s_i)=3$, $1\le j\le 5,i\in\{1,2\}$, and therefore there exists a vertex $w'\in K$, $w'\notin \{w_1,\ldots,w_5,s_1,s_2\}$ such that $w'v_3\in E(G)$.  Finally, the path $P_1=(w',v_3,v,v_2,w_1\overrightarrow{P_a}w_5,v_1,s_1,t_1,s_2)$ is a desired path. \bb{Now we shall explore the case in which $w_5v_2\in E(G)$. If $w_5v_2\in E(G)$, then $w_5x_2\in E(G)$, otherwise either \NC{w_2}{w_5} or \NC{w_3}{w_5} has an induced \K14.  Further, since $w_5v_2\in E(G)$, note that $v_1s_i,v_2s_i\in E(G)$, $i\in \{1,2\}$.}  Now $P_2=(w',v_3,v,v_2,s_1,t_1,s_2,v_1,w_4,x_4,w_5,x_2,w_3,x_3,w_1,x_1,w_2)$ is a desired path.  Since the paths \bb{$P_1$ and $P_2$ have} its end vertices in $K$, using clique edges we obtain a Hamiltonian cycle in $G$. \bb{If $w_5v_3\in E(G)$, then note that for some $i$, $1\le i\le 3$, \NC{w_i}{w_5} has an induced \K14.}  
\item \emph{\bf {\em Case} 1.2:}  $w_1x_3,w_3v_2\in E(G)$.
Similar to the previous case, either $v_1s_i,v_2s_i\in E(G)$ or  $v_1s_i,x_3s_i\in E(G)$, $i\in \{1,2\}$ and thus $v_1s_i\in E(G)$.  Further, similar to \bb{the} previous case, there exists $w'\in K$ such that $w'v_3\in E(G)$.  Now if $w_5v_1\in E(G)$, then the path $P_1$ as obtained in the previous case is a desired path.  If $w_5v_1\nin E(G)$, then either $w_5v_2,w_5x_1\in E(G)$ or $w_5v_2,w_5x_2\in E(G)$.  Thus $(w',v_3,v,v_2,w_1,x_1,w_5\overleftarrow{P_a}w_2,v_1,s_1,t_1,s_2)$ or   
$(w',v_3,v,v_2,w_3,x_2,w_5,x_4,w_4,x_3,w_1,x_1,w_2,v_1,s_1,t_1,s_2)$ is a desired path. \bb{Note $w_5v_3\notin E(G)$ as in the previous case.}
\item \emph{\bf {\em Case} 1.3:}  $w_1x_4,w_3v_2\in E(G)$.
Note that $v_1s_i,v_2s_i\in E(G)$, $i\in \{1,2\}$.  Further, similar to the previous case, there exists $w'\in K$ such that $w'v_3\in E(G)$.  Now if $w_5v_1\in E(G)$, then the path $P_1$ as obtained in Case 1.1 is a desired path.  \bb{Note $w_5v_3\notin E(G)$ as in the Case 1.1}.  Therefore, if $w_5v_1\nin E(G)$, then $w_5v_2\in E(G)$.  We obtain $(w',v_3,v,v_2,w_5,x_4,w_1\overrightarrow{P_a}w_4,v_1,s_1,t_1,s_2)$ 
as a desired path. 
\item \emph{\bf {\em Case} 1.4:}  $w_1x_3,w_3v_3\in E(G)$.
Note that $v_1s_i,x_3s_i\in E(G)$, $i\in \{1,2\}$, \bb{otherwise for some $y'\in \{w_1\ldots w_5,v\}$, \NC{y'}{s_i} has an induced \K14}.  If $w_5$ is adjacent to $v_1,x_3$, then consider $S=\{v_2,v_3,x_1,x_2\}$.  Further, for every $z'\in S$, $d^I(z')=2$, and $S\cup N(S)$ has a short cycle.  Since $G$ has no short cycles, there exists a vertex $w'\in K$ such that $w'$ is adjacent to some vertices in $S$.  In this case one of the following is a desired path.\\
$(w',v_2,v,v_3,w_3,x_2,w_2,x_1,w_1,x_3,w_4,x_4,w_5,v_1,s_1,t_1,s_2)$ \\
$(w',v_3,v,v_2,w_1\overrightarrow{P_a}w_5,v_1,s_1,t_1,s_2)$ \\
$(w',x_1\overrightarrow{P_a}w_3,v_3,v,v_2,w_1,x_3\overrightarrow{P_a}w_5,v_1,s_1,t_1,s_2)$ \\ 
$(w',x_2,w_2,x_1,w_1,v_2,v,v_3,w_3\overrightarrow{P_a}w_5,v_1,s_1,t_1,s_2)$ \\  
Now we see the case in which $w_5$ is adjacent to $v_2$ or $v_3$.  
\bb{If $\mathbb{C}=P_a\cup P_b$, then} note that \\ $C_1=(s_2,t_1,s_1,v_1,w_4,x_4,w_5,v_2,v,v_3,w_3,x_2,w_2,x_1,w_1,x_3,s_2)$ or \\
$C_2=(s_2,t_1,s_1,v_1,w_4,x_4,w_5,v_3,v,v_2,w_1\overrightarrow{P_a}x_3,s_2)$  are possible Hamiltonian cycles. \bb{If there exists $P_c\in \mathbb{C}$ other than $P_a$ and $P_b$, then depending on the adjacency of the end vertex of $P_c$ with $N^I(v)$, we obtain the following desired paths. One among $(\overrightarrow{P_c},v_1\overrightarrow{C_1}s_1)$, $(\overrightarrow{P_c},v_2\overrightarrow{C_1}w_5)$, $(\overrightarrow{P_c},v_3\overrightarrow{C_1}v)$, $(\overrightarrow{P_c},v_1\overrightarrow{C_2}s_1)$, $(\overrightarrow{P_c},v_2\overrightarrow{C_2}v)$, $(\overrightarrow{P_c},v_3\overrightarrow{C_2}w_5)$,  is a desired path.} 
\end{enumerate} 
\item  \emph{\bf {\em Case} 2:}  $v_1w_3\notin E(G)$. 
In this case we shall assume that $v_1w_i\in E(G)$, $i\in \{1,2,4,5\}$.  
Note that $w_3v_2\in E(G)$ or $w_3v_3\in E(G)$, say $w_3v_2\in E(G)$.  Clearly, $s_iv_1,s_iv_2\in E(G)$ or $s_iv_1,s_ix_2\in E(G)$  or $s_iv_1,s_ix_3\in E(G)$, $i\in \{1,2\}$.  Similarly, $w_1,w_5$ are adjacent to one of $v_2,x_2,x_3$.  Therefore there exists $w'\in K$ such that $w'v_3\in E(G)$.  We see the following cases depending on the adjacency of $s_i,w_1,w_5$ with $N^I(w_3)$.  
\begin{enumerate}[  ]
\item \emph{\bf {\em Case} 2.1:} $w_1v_2\in E(G)$ or $w_5v_2\in E(G)$, \bb{without loss of generality,} let $w_1v_2\in E(G)$. \\
$(w',v_3,v,v_2,w_1\overrightarrow{P_a}w_5,v_1,\overrightarrow{P_b})$ is a desired path.
\item \emph{\bf {\em Case} 2.2:} $w_1x_3\in E(G)$ or $w_5x_2\in E(G)$.
$(w',v_3,v,v_2,w_3,x_2\overleftarrow{P_a}w_1,x_3\overrightarrow{P_a}w_5,v_1,\overrightarrow{P_b})$ or \\
$(w',v_3,v,v_2,w_3,x_3\overrightarrow{P_a}w_5,x_2\overleftarrow{P_a}w_1,v_1,\overrightarrow{P_b})$ is a desired path.
\item \emph{\bf {\em Case} 2.3:} \bb{$w_1x_2,w_5x_3\in E(G)$}.
If $s_iv_2\in E(G)$, say $s_1v_2\in E(G)$, then $(w',v_3,v,v_2,s_1,t_1,s_2,v_1,w_1\overrightarrow{P_a}w_5)$ is a desired path.
If $s_ix_2\in E(G)$, say $s_1x_2\in E(G)$, then $(w',v_3,v,v_2,w_3\overrightarrow{P_a}w_5,v_1,w_1\overrightarrow{P_a}x_2,s_1,t_1,s_2)$ is a desired path.
If $s_ix_3\in E(G)$, say $s_1x_3\in E(G)$, then $(w',v_3,v,v_2,w_3\overleftarrow{P_a}w_1,v_1,w_5\overleftarrow{P_a}x_3,s_1,t_1,s_2)$ is a desired path.
\end{enumerate} 
\item \emph{\bf {\em Case} 3:}  $v_1w'\notin E(G)$, $w'\nin P_a$.
In this case \bb{note} that $v_1w_i\in E(G)$, $1\le i\le 5$.  
Note that in this case $w'\nin P_b$.  Suppose $w'\in P_b$, say $v_1s_1\nin E(G)$, then $s_1$ is adjacent to at least two vertices in $x_1,\ldots,x_4$.  Further, either $s_1v_2\in E(G)$ or $s_1v_3\in E(G)$.  Thus \NV{s_1} induces a \K14, a contradiction.  Therefore $w'\nin P_b$.  It follows that $w'$ is an end vertex of some path $P_c\in \mathbb{P}_1$.  Note that $w'v_2\in E(G)$ or $w'v_3\in E(G)$, say $w'v_2\in E(G)$. 
If $v_3$ is adjacent to $P_a$, then $(\overrightarrow{P_b},v_1,\overrightarrow{P_a},v_3,v,v_2,w'\overrightarrow{P_c})$ is a desired path.  If $v_3$ is adjacent to $P_d\nin \{P_a,P_b,P_c\}$, then $(\overrightarrow{P_b},v_1,\overrightarrow{P_a},\overrightarrow{P_d},v_3,v,v_2,w'\overrightarrow{P_c})$ is a desired path.  In all of the above cases, the end vertices of the desired paths are in $K$, and thus using clique edges we get a Hamiltonian cycle in $G$.
\end{enumerate}
This completes the case analysis and a proof. $\hfill\qed$
\end{proof}  
\begin{cl}\label{no775} 
  $|\mathbb{P}_7|\le2$.  Further, if $|\mathbb{P}_7|=2$, then $\mathbb{P}_5=\emptyset$. 
\end{cl}
\begin{proof}
We first show that $|\mathbb{P}_7|\le2$.  Suppose that there exists $P_a,P_b,P_c\in \mathbb{P}_7$ such that  $P_a=(w_1,\ldots,w_4;$ $x_1,\ldots,x_3)$, $P_b=(s_1,\ldots,s_4;t_1,\ldots,t_{3})$, and $P_c=(q_1,\ldots,q_4;r_1,\ldots,r_3)$. 
By Claim \ref{xunivtwopaths}, there exists a vertex $v_1\in N^I(v)$ such that $v_1w_j,v_1s_j,v_1q_j\in E(G), j\in\{2,3\}$.  Since the clique is maximal there exists $w'\in K$ such that $w'v_1\notin E(G)$.  It follows that $w'v_2\in E(G)$ or $w'v_3\in E(G)$.  \bb{Further, $w'$ is adjacent to at least one of $\{x_1,x_2,x_3\}$, and at least one of $\{t_1,t_2,t_3\}$, and at least one of $\{r_1,r_2,r_3\}$, otherwise for some $i\in \{2,3\}$, \NC{w_i}{w'} or \NC{s_i}{w'} or \NC{q_i}{w'} has an induced $K_{1,4}$.}  Now $\{w'\}\cup N^I(w')$ induces a $K_{1,4}$, a contradiction to the existence of three such paths $P_a,P_b,P_c$.
To prove the second half, let $P_a,P_b\in \mathbb{P}_7$.   For a contradiction, assume that $P_d\in \mathbb{P}_5$ such that $P_d=(y_1,y_2,y_3;z_1,z_2)$, 
From Claim \ref{xunivtwopaths}, there exists $v_1\in N^I(v)$ such that $v_1w_i,v_1s_i,v_1y_2\in E(G)$, $i\in \{2,3\}$.  Now we claim that $v_1w_i,v_1s_i,v_1y_j\in E(G)$, $i\in\{1,4\}, j\in\{1,3\}$.  Suppose $v_1w_1\notin E(G)$, then by Claim A, either $v_2w_1\in E(G)$ or $v_3w_1\in E(G)$.  \bb{ Further, $w_1$ is adjacent to at least one of $\{t_1,t_2,t_3\}$, and at least one of $\{z_1,z_2\}$, otherwise for some $i\in \{2,3\}$, \NC{s_i}{w_1} or \NC{y_2}{w_1} has an induced $K_{1,4}$. }  
It follows that $\{w_1\}\cup N^I(w_1)$ induces a $K_{1,4}$, a contradiction to the assumption that $v_1w_1\notin E(G)$.  
Similar arguments hold good for the other edges.
Finally, for any arbitrary vertex $w'\in K$ where $w'\notin \{w_1,\ldots,w_4,s_1,\ldots,s_4,y_1,\ldots,y_3\}$, we claim $v_1w'\in E(G)$.  If not, then by Claim A, either $v_2w'\in E(G)$ or $v_3w'\in E(G)$.  Further, similar to the previous argument for the vertex $w_1$, the same arguments hold good for $w'$, i.e., $w'$ is adjacent to at least one of $\{x_1,x_2,x_3\}$, and at least one of $\{t_1,t_2,t_3\}$, and at least one of $\{z_1,z_2\}$.  Now, $\{w'\}\cup N^I(w')$ induces a $K_{1,4}$, a contradiction.  Thus $v_1w'\in E(G)$.  Therefore, we conclude that $\{v_1\}\cup K$ is a clique of larger size, which is the final contradiction to the existence of $P_d$.  This completes the proof.   $\hfill\qed$
\end{proof}
\begin{cl}\label{no755} 
 If $|\mathbb{P}_7|=1$, then $|\mathbb{P}_5|\le1$. 
\end{cl}
\begin{proof}
Let $P_a\in\mathbb{P}_7$, where $P_a=(w_1,\ldots,w_4;x_1,\ldots,x_3)$.  Assume for a contradiction that there exists paths $P_b,P_c\in \mathbb{P}_5$ such that  $P_b=(s_1,s_2,s_3;t_1,t_2)$, and $P_c=(q_1,q_2,q_3;r_1,r_2)$.  
From Claim \ref{xunivtwopaths}, there exists $v_1\in N^I(v)$ such that $v_1w_2,v_1w_3,v_1s_2,v_1q_2\in E(G)$.  Similar to the proof of previous claim we could argue that $v_1w_1,v_1w_4,v_1s_i,v_1q_i\in E(G), i\in \{1,3\}$.  Further, for any arbitrary vertex $w'\in K$ such that $w'\nin \{w_1,\ldots,w_4,s_1,s_2,s_3,q_1,q_2,q_3\}$, we claim $v_1w'\in E(G)$.  Suppose not, then by Claim A, either $v_2w'\in E(G)$ or $v_3w'\in E(G)$.  Further, similar to the proof of previous claim, we could argue that \bb{$w'$ is adjacent to at least one of $\{x_1,x_2,x_3\}$, and at least one of $\{t_1,t_2\}$, and at least one of $\{r_1,r_2\}$.}  Now, $\{w'\}\cup N^I(w')$ induces a $K_{1,4}$, which is a contradiction to the assumption that $v_1w'\nin E(G)$.  Therefore, we conclude that $\{v_1\}\cup K$ is a clique of larger size, which is the final contradiction to the existence of two such paths $P_b,P_c$.  This completes the proof.   $\hfill\qed$
\end{proof}
\begin{cl}\label{HC771} 
If there exist $P_a,P_b\in\mathbb{P}_7$, then $G$ has a Hamiltonian cycle. 
\end{cl}
\begin{proof}
Let $P_a,P_b\in\mathbb{P}_7$ such that  $P_a=(w_1,\ldots,w_4;x_1,\ldots,x_3)$, $P_b=(s_1,\ldots,s_4;t_1,$ $\ldots,t_{3})$.  Similar to the arguments in the proof of Claim \ref{no775}, there exists $v_1\in N^I(v)$ such that $v_1w_i,v_1s_i\in E(G),1\le i\le4$.  Since $K$ is a maximal clique, there exists $w'\in K$ such that $w'v_1\nin E(G)$.  From Claim A, either $w'v_2\in E(G)$ or $w'v_3\in E(G)$.  Without loss of generality, let $w'v_3\in E(G)$.  Note that $w'x_2\in E(G)$, otherwise, either \NC{w_2}{w'} induces a \K14 or \NC{w_3}{w'} induces a \K14.  Similarly, $w't_2\in E(G)$.  Note that the vertices $w_i,s_i, i\in \{1,4\}$ is adjacent to one of the vertices in $\{v_3,t_2,x_2\}$, if not, say $w_1v_3,w_1t_2,w_1x_2\nin E(G)$, then \NC{w'}{w_1} induces a \K14.  Similar arguments hold for $w_2,s_1,s_2$.  It follows that for every $s'\in S=\{w',w_1\ldots,w_4,s_1,\ldots,s_4\}$, $d^I(s')=3$ \bb{and $v_2 \notin N^I(s')$}.  Since $G$ is two connected, there exists $w''\in K\setminus S$ such that $w''v_2\in E(G)$.  Observe that $(w'',v_2,v,v_3,w',w_1\overrightarrow{P_a}w_4,v_1,s_1\overrightarrow{P_b}s_4)$ is a \bb{desired} path containing $N^I(v)$ which could be easily extended to a Hamiltonian cycle in $G$.  $\hfill\qed$ 
\end{proof} 
\begin{cl}\label{HC75}
If there exist $P_a\in\mathbb{P}_7,P_b\in\mathbb{P}_5$ and $G$ has no short cycle, then $G$ has a Hamiltonian cycle.
\end{cl}
\begin{proof}
Let $P_a\in\mathbb{P}_7$, $P_b\in \mathbb{P}_5$, such that  $P_a=(w_1,\ldots,w_4;x_1,\ldots,x_3)$, $P_b=(s_1,\ldots,s_3;t_1,t_{2})$. 
From Claim \ref{xunivtwopaths}, there exists $v_1\in N^I(v)$ such that $v_1w_i,v_1s_2\in E(G),i\in\{2,3\}$.    
Now we claim that $v_1w_1\in E(G)$.  If not, observe that either $v_2w_1\in E(G)$ or $v_3w_1\in E(G)$.  Also note that $w_1x_2\in E(G)$ or $w_1x_3\in E(G)$, otherwise \NC{w_3}{w_1} induces a \K14.  Further, $w_1t_1\in E(G)$ or $w_1t_2\in E(G)$, otherwise \NC{s_2}{w_1} induces a \K14.  Clearly, \NV{w_1} induces a \K14, a contradiction.  Therefore, $v_1w_1\in E(G)$.  Similarly, $v_1w_4\in E(G)$.  
Since the clique is maximal, there exists $w'\in K$ such that $v_1w'\nin E(G)$.  We see the following cases.
\begin{enumerate}[  ]
\item \emph{\bf {\em Case} 1:} $v_1s_1,v_1s_3\in E(G)$, \bb{and therefore, $w'\nin \{s_1,s_3\}$. } 
From the previous claims, it is easy to see that $w'$ is an end vertex of a path $P_c$ in $\mathbb{P}_1$.  
From Claim A, \bb{$w'$ is adjacent to a vertex in $N^I(v)$. Without loss of generality,} $v_3w'\in E(G)$.  Now we claim $w'x_2\in E(G)$, otherwise, either \NC{w_2}{w'} induces a \K14 or \NC{w_3}{w'} induces a \K14.
Also observe that either $w't_1$ or $w't_2$ is in $E(G)$, otherwise \NC{s_2}{w'} induces a \K14. Without loss of generality, let $w't_2\in E(G)$. 
\bb{Now note that any vertices $z\in \{w_1,w_4,s_1\}$ has an adjacency in $\{t_2,x_2,v_3\}$, otherwise \NC{w'}{z} has an induced \K14}.  Clearly, $d^I(w_j)=d^I(s_1)=d^I(s_2)=3$, $1\leq j\leq 4$ and since the graph is $2$-connected, $v_2s_3\in E(G)$ or $v_2w''\in E(G)$ where $w''$ is the end vertex of a path $P_d$ in $\mathbb{P}_3\cup \mathbb{P}_1$.  Here we obtain   $(\overrightarrow{P_c}w',v_3,v,v_2,s_3\overleftarrow{P_b}s_1,v_1,w_1\overrightarrow{P_a}w_4)$ or
$(\overrightarrow{P_d}w'',v_2,v,v_3,w'\overrightarrow{P_c},w_1\overrightarrow{P_a}w_4,v_1,s_1\overrightarrow{P_b}s_3)$ as a desired path.
\item \emph{\bf {\em Case} 2:} \bb{$v_1s_1,v_1s_3\notin E(G)$, and therefore, } $w'\in \{s_1,s_3\}$.  
Without loss of generality, let $w'=s_1$, i.e., $v_1s_1\nin E(G)$. 
From Claim A, \bb{$s_1$ is adjacent to a vertex in $N^I(v)$. Without loss of generality,}  $v_3s_1\in E(G)$.  Also note that $s_1x_2\in E(G)$, otherwise either \NC{w_2}{s_1} induces a \K14 or \NC{w_3}{s_1} induces a \K14.
Now we claim that $w_1$ and $w_4$ are adjacent to one of the vertices in $\{v_3,t_1,x_2\}$.  Suppose $w_1v_3,w_1t_1,w_1x_2\nin E(G)$, then \NC{s_1}{w_1} induces a \K14.  Similar arguments hold for $w_4$.  It follows that $d^I(w_j)=d^I(s_k)=3$, $1\leq j\leq4$, $k\in \{1,2\}$.  Since $G$ is 2-connected, there exists $w^*\in K$ such that $w^*v_2\in E(G)$.  We see the following sub cases based on the possibility of $w^*$.
\begin{enumerate}[  ]
\item \emph{\bf {\em Case} 2.1:} $w^*=s_3$.  i.e., $v_2s_3\in E(G)$.
In this sub case we claim that there exists a vertex $w''\neq v\in K$ such that $w''\nin P_a\cup P_b$ and $w''x'\in E(G)$ where $x'\in\{v_2,v_3,t_1,t_2\}$.   Suppose such a $w''$ does not exist, then observe that, in the set $S=\{t_1,t_2,v_2,v_3\}$, $d(t_1)=d(t_2)=d(v_2)=d(v_3)=2$, and $S\cup N(S)$ has a short cycle, a contradiction.  Note that, $w''$ is an end vertex of a path $P_d$ in $\mathbb{P}_3\cup \mathbb{P}_1$.   
Now depending on the adjacency of $w''$, we obtain the following paths.\\
If $w''v_2\in E(G)$, then  we obtain $(\overrightarrow{P_d}w'',v_2,s_3\overleftarrow{P_b}s_1,v_3,v,v_1,w_1\overrightarrow{P_a}w_4)$ as a desired path.  \\
If $w''v_3\in E(G)$, then  we obtain $(\overrightarrow{P_d}w'',v_3,s_1\overrightarrow{P_b}s_3,v_2,v,v_1,w_1\overrightarrow{P_a}w_4)$ as a desired path.  \\
If $w''t_1\in E(G)$, then  we obtain $(\overrightarrow{P_d}w'',t_1,s_1,v_3,v,v_2,s_3\overleftarrow{P_b}s_2,v_1,w_1\overrightarrow{P_a}w_4)$ as a desired path.  \\
If $w''t_2\in E(G)$, then  we obtain $(\overrightarrow{P_d}w'',t_2,s_3,v_2,v,v_3,s_1\overrightarrow{P_b}s_2,v_1,w_1\overrightarrow{P_a}w_4)$ as a desired path.  
\item \emph{\bf {\em Case} 2.2:}  $w^*\neq s_3$.
Note that $w^*$ is an end vertex of a path $P_e$ in $\mathbb{P}_3\cup \mathbb{P}_1$.   
We see the following sub cases to complete our argument.\\
\emph{\bf {\em Case} 2.2.1:} $v_1s_3\in E(G)$.
Here we obtain $(\overrightarrow{P_e}w^*,v_2,v,v_3,s_1\overrightarrow{P_b}s_3,v_1,w_1\overrightarrow{P_a}w_4)$ as a desired path.  \\
\emph{\bf {\em Case} 2.2.2:} $v_1s_3\nin E(G)$.
Clearly, from Claim A either $s_3v_2\in E(G)$ or $s_3v_3\in E(G)$.  We now claim that $s_3x_2\in E(G)$.  Otherwise either \NC{w_2}{s_3} induces a \K14 or \NC{w_3}{s_3} induces a \K14. 
Here we obtain $(\overrightarrow{P_e}w^*,v_2,v,v_3,s_1\overrightarrow{P_b}s_3,x_2\overrightarrow{P_a}w_4,v_1,w_2,x_1,w_1)$ as a desired path.
\end{enumerate}  
\end{enumerate}
This completes the case analysis and the proof.  $\hfill\qed$
\end{proof} 
\begin{cl}\label{HC733}
If there exist $P_a\in\mathbb{P}_7,P_b,P_c\in\mathbb{P}_3$ and $\mathbb{P}_5=\emptyset$ and $G$ has no short cycles, then $G$ has a Hamiltonian cycle.  
\end{cl}
\begin{proof}
Let $P_a=(w_1,\ldots,w_4;,x_1,\ldots,x_3)$, $P_b=(s_1,s_2;t_1)$, and $P_c=(q_1,q_2;r_1)$.     
From Claim A, the vertices $w_2,w_3$ are adjacent to at least one of the vertices in $N^I(v)$.  Depending on this adjacency, we see the following two cases.
\begin{enumerate}[  ]
\item \emph{\bf {\em Case} 1:} \bb{$w_2$ and $w_3$ are adjacent to the same vertex in $N^I(v)$. Without loss of generality,} let $v_1\in N^I(v)$ such that $v_1w_2,v_1w_3\in E(G)$. \bb{Let $S=\{s_1,s_2,q_1,q_2\}$.}
We first claim that \bb{all} the vertices in $S$ are adjacent to at least one of $v_1,x_2$.
Suppose that \bb{$s_1$ is not adjacent to both, i.e.,} $v_1s_1,x_2s_1\nin E(G)$.   Note that $v_2s_1\in E(G)$ or $v_3s_1\in E(G)$.  If $s_1x_1\nin E(G)$, then \NC{w_2}{s_1} induces a \K14.  Therefore,  $s_1x_1\in E(G)$ and similarly, $s_1x_3\in E(G)$, otherwise \NC{w_3}{s_1} induces a \K14.  It follows that \NV{s_1} induces a \K14.  Similar argument holds for other vertices, and thus the vertices $s_1,s_2,q_1,q_2$ are adjacent to at least one of $v_1,x_2$.  Since the clique is maximal, there exists $v'\in K$ such that $v_1v'\nin E(G)$.  We further classify based on the possibilities of $v'$ as follows.
\begin{enumerate}[  ]
\item
\emph{\bf {\em Case} 1.1:} $v'\in P_a$, i.e., $w_1$ or $w_4$ is non-adjacent to $v_1$, say $v_1w_4\nin E(G)$. \\  
Note that either $v_2w_4\in E(G)$ or $v_3w_4\in E(G)$.  Without loss of generality, let $v_2w_4\in E(G)$.  
\bb{Since $G$ is 2-connected,} note that $v_3$ is adjacent to at least one more clique vertex, say $v_3v''\in E(G)$, where $v''$ is an end vertex of some path $P_d$.  Further, we see the following cases depending on the possibilities of $v''$.
\begin{enumerate}[  ]
\item  \emph{\bf {\em Case} 1.1.1:} $v''\nin P_a\cup P_b\cup P_c$.  Then depending on the adjacency of $s_1,s_2,q_1,q_2$ to $v_1,x_2$, one of the following is a desired path $P$.\\
If all of $s_i,q_i$ are adjacent to $v_1$, then $P=(\overrightarrow{P_d}v'',v_3,v,v_2,w_4\overleftarrow{P_a}w_1, s_1,t_1,s_2,v_1,q_1,r_1,q_2)$\\
If all of $s_i,q_i$ are adjacent to $x_2$, then $P=(\overrightarrow{P_d}v'',v_3,v,v_2,w_4\overleftarrow{P_a}w_3,v_1,w_2\overleftarrow{P_a}w_1,s_1,t_1,s_2,x_2,q_1,r_1,q_2)$\\
If $s_i,q_i$ are adjacent to different vertices in $v_1,x_2$, say $v_1s_1,q_2x_2\in E(G)$, then\\ $P=(\overrightarrow{P_d}v'',v_3,v,v_2,w_4\overleftarrow{P_a}w_3,v_1,s_1,t_1,s_2,q_1,r_1,q_2,x_2,w_2\overleftarrow{P_a}w_1)$.
\item   \emph{\bf {\em Case} 1.1.2:} $v''\in P_a$. \\
Note that either $w_4x_1\in E(G)$ or $w_4x_2\in E(G)$, otherwise \NC{w_2}{w_4} has an induced \K14.  Therefore, $v''=w_1$, i.e., $v_3w_1\in E(G)$.  Now we see the following sub cases.  \bb{Recall that all the vertices in $S$ are adjacent to at least one of $v_1,x_2$.}
If at least three vertices in $S$ are having adjacency with $v_1$, say $s_1v_1,q_1v_1,q_2v_1\in E(G)$.  Note that all the vertices $s_1,q_1,q_2$ are having adjacency with one vertex in $N^I(w_4)$.  We obtain desired path $P$ as follows.
If $q_2v_2\in E(G)$, then $P=(s_2,t_1,s_1,v_1,q_1,r_1,q_2,v_2,w_4\overleftarrow{P_a}w_1,v_3,v)$. \\
If $q_2x_3\in E(G)$, then $P=(s_2,t_1,s_1,v_1,q_1,r_1,q_2,x_3\overleftarrow{P_a}w_1,v_3,v,v_2,w_4)$.\\
If $q_2x_2\in E(G)$, then $P=(s_2,t_1,s_1,v_1,q_1,r_1,q_2,x_2\overrightarrow{P_a}w_4,v_2,v,v_3,w_1,x_1,w_2)$.\\
If $q_2x_1\in E(G)$, then $P=(s_2,t_1,s_1,v_1,q_1,r_1,q_2,x_1\overrightarrow{P_a}w_4,v_2,v,v_3,w_1)$.\\
Similarly, if at least three vertices in $S$ are having adjacency with $x_2$, say $s_1x_2,q_1x_2,q_2x_2\in E(G)$.  Note that all the vertices $s_1,q_1,q_2$ are having adjacency with one vertex in $\{v_2,v_3\}$.  Further, we obtain $P$ as follows.
If $q_2v_2\in E(G)$, then $P=(s_2,t_1,s_1,x_2,q_1,r_1,q_2,v_2,w_4,x_3,w_3,v_1,w_2,x_1,w_1,v_3,v)$. \\
If $q_2v_3\in E(G)$, then \bb{$P=(s_2,t_1,s_1,x_2,q_1,r_1,q_2,v_3,v,v_2,w_4,x_3,w_3,v_1,w_2,x_1,w_1)$.} \\
Now we see the case in which exactly two vertices in $S$ are adjacent to $x_2$ \bb{and the remaining two to} $v_1$.  If $s_1v_1,q_1v_1\in E(G)$ and $s_2x_2,q_2x_2\in E(G)$, then \bb{note that $q_2\cap N^I(v)\neq \emptyset$.}  We obtain $P$ as follows.\\
If $q_2v_2\in E(G)$, then $P=(s_2,t_1,s_1,v_1,q_1,r_1,q_2,v_2,w_4\overleftarrow{P_a}w_1,v_3,v)$.\\
If $q_2v_3\in E(G)$, then $P=(s_2,t_1,s_1,v_1,q_1,r_1,q_2,v_3,v,v_2,w_4\overleftarrow{P_a}w_1)$. \\
If $s_1v_1,s_2v_1\in E(G)$ and $q_1x_2,q_2x_2\in E(G)$, then we obtain \\ \bb{$P=(s_2,t_1,s_1,v_1,w_2,x_1,w_1,v_3,v,v_2,w_4,x_3,w_3,x_2,q_1,r_1,q_2)$.}  The other cases are symmetric, when \\ $s_1x_2,q_1x_2,s_2v_1,q_2v_1\in E(G)$ and $q_1v_1,q_2v_1,s_1x_2,s_2x_2\in E(G)$.  This completes Case $1.1.2$. 
\item   \emph{\bf {\em Case} 1.1.3:} $v''\in P_b\cup P_c$. \\
Without loss of generality, let $v''=s_1$, i.e., $s_1v_3\in E(G)$.
Consider the adjacency of the vertices $s_2,q_1$ with $v_1,x_2$.   If $v_1s_2,v_1q_1\in E(G)$, then the desired path $P=(q_2,r_1,q_1,v_1,s_2,t_1,s_1,v_3,v,v_2,w_4\overleftarrow{P_a}w_1)$.  If $x_2s_2,x_2q_1\in E(G)$, then  $P=(q_2,r_1,q_1,x_2,s_2,t_1,s_1,v_3,v,v_2,w_4\overleftarrow{P_a}w_3,v_1,w_2\overleftarrow{P_a}w_1)$.  
If $v_1s_2,x_2q_1\in E(G)$, then  \bb{$P=(w_1\overrightarrow{P_a}w_2,v_1,s_2,t_1,s_1,v_3,v,v_2,w_4\overleftarrow{P_a}x_2,q_1,r_1,q_2)$.}  The other case is symmetric when $x_2s_2,v_1q_1\in E(G)$.
%
%
%
\end{enumerate}
\item  \emph{\bf {\em Case} 1.2:} $v'\in \mathbb{P}_3$, without loss of generality $v_1s_1\nin E(G)$. \\
Note that $v_1w_i\in E(G)$, $1\le i\le 4$.   
Note that either $s_1v_2\in E(G)$ or $s_1v_3\in E(G)$, say $s_1v_2\in E(G)$.  Since $s_1v_1\nin E(G)$, clearly $s_1x_2\in E(G)$, \bb{otherwise \NC{w_2}{s_1} or \NC{w_3}{s_1} has an induced \K14.}  
Since $G$ is 2-connected, there exists $z'\in K$ such that $v_3z'\in E(G)$.  
Since there exists an adjacency for $w_1,w_4$ in $N^I(s_1)$, $d^I(w_i)=3$, and thus $z'\ne w_i$, $1\le i\le 4$.
If $z'\nin P_a\cup P_b$, then $z'$ is an end vertex of some path $P_d\in \mathbb{P}_3\cup \mathbb{P}_1$. \bb{Recall that all the vertices in $S$ are adjacent to at least one of $v_1,x_2$.} We obtain  $(\overrightarrow{P_d}z',v_3,v,v_2,s_1,t_1,s_2,v_1,w_1\overrightarrow{P_a}w_4,\overrightarrow{P_c})$ or $(\overrightarrow{P_d}z',v_3,v,v_2,s_1,t_1,s_2,x_2,w_3,x_3,w_4,v_1,w_1,x_1,w_2,\overrightarrow{P_c})$ as the desired path.
If $z'\in P_c$, say $v_3q_1\in E(G)$, then 
$(\overleftarrow{P_c}q_1,v_3,v,v_2,s_1,t_1,s_2,v_1,w_1\overrightarrow{P_a}w_4)$ or \\ $(\overleftarrow{P_c}q_1,v_3,v,v_2,s_1,t_1,s_2,x_2,w_3,x_3,w_4,v_1,w_1,x_1,w_2)$ is the desired path.
Now the left over case is when $z'=s_2$.
Observe that $d(v_3)=d(v_2)=d(t_1)=2$. Let \bb{$R=\{v_2,v_3,t_1\}$. Then $R\cup N(R)$} induces a short cycle.  Since $G$ has no short cycles, there exists some more adjacency for vertices in $R$.  Observe that both $w_1,w_4$ have adjacency to a vertex in $N^I(s_1)$.
Now if any one of $w_1,w_4$ is adjacent to any one of  $t_1,v_2$, then we obtain the desired path $P$ as follows.\\
If $w_1v_2\in E(G)$ and $q_1v_1\in E(G)$, then $P=(q_2,r_1,q_1,v_1,w_4\overleftarrow{P_a}w_1,v_2,v,v_3,s_2,t_1,s_1)$.\\
If $w_1v_2\in E(G)$ and $q_1x_2\in E(G)$, then $P=(q_2,r_1,q_1,x_2\overrightarrow{P_a}w_4,v_1,w_2,x_1,w_1,v_2,v,v_3,s_2,t_1,s_1)$.\\
If $w_4v_2\in E(G)$ and $q_1v_1\in E(G)$, then $P=(q_2,r_1,q_1,v_1,w_1\overrightarrow{P_a}w_4,v_2,v,v_3,s_2,t_1,s_1)$.\\
If $w_4v_2\in E(G)$ and $q_1x_2\in E(G)$, then $P=(q_2,r_1,q_1,x_2\overleftarrow{P_a}w_1,v_1,w_3,x_3,w_4,v_2,v,v_3,s_2,t_1,s_1)$.\\
If $w_1t_1\in E(G)$ and $q_1v_1\in E(G)$, then $P=(q_2,r_1,q_1,v_1,w_4\overleftarrow{P_a}w_1,t_1,s_2,v_3,v,v_2,s_1)$.\\
If $w_1t_1\in E(G)$ and $q_1x_2\in E(G)$, then $P=(q_2,r_1,q_1,x_2\overrightarrow{P_a}w_4,v_1,w_2,x_1,w_1,t_1,s_2,v_3,v,v_2,s_1)$.\\
If $w_4t_1\in E(G)$ and $q_1v_1\in E(G)$, then $P=(q_2,r_1,q_1,v_1,w_1\overrightarrow{P_a}w_4,t_1,s_2,v_3,v,v_2,s_1)$.\\
If $w_4t_1\in E(G)$ and $q_1x_2\in E(G)$, then $P=(q_2,r_1,q_1,x_2\overleftarrow{P_a}w_1,v_1,w_3,x_3,w_4,t_1,s_2,v_3,v,v_2,s_1)$.\\
Now we shall consider the case where $w_1x_2,w_4x_2\in E(G)$.  Note that either $q_1v_1\in E(G)$ or $q_1v_2\in E(G)$  or $q_1v_3\in E(G)$.  If one of $q_1,q_2$ is adjacent to one of $v_2,v_3$, then we obtain $P$ as follows.  If $q_1v_2\in E(G)$, then $P=(q_2,r_1,q_1,v_2,s_1,t_1,s_2,v_3,v,v_1,w_1\overrightarrow{P_a}w_4)$.  \\
If $q_1v_3\in E(G)$, then $P=(q_2,r_1,q_1,v_3,s_2,t_1,s_1,v_2,v,v_1,w_1\overrightarrow{P_a}w_4)$.  \\
Now if $q_iv_1\in E(G)$, $i\in \{1,2\}$, then $q_i$ has an adjacency in $N^I(s_1)$.  Further, if $q_1t_1\in E(G)$, then $P=(q_2,r_1,q_1,t_1,s_2,v_3,v,v_2,s_1,x_2\overrightarrow{P_a}w_4,v_1,w_2,x_1,w_1)$.\\
Finally, we are left with one case that $q_ix_2\in E(G)$.  We claim that such a case cannot occur.  Suppose not, then observe that $d^I(w_i)=d^I(s_j)=d^I(q_j)=3$, $i\in\{1,2,3\}, j\in \{1,2\}$, and $d(v_2)=d(v_3)=d(t_1)=2$.  Let $R=\{v_2,v_3,t_1\}$ then $R\cup N(R)$ has a short cycle, a contradiction.  
\item  \emph{\bf {\em Case} 1.3:}  $v' \nin P_a\cup \mathbb{P}_3$.\\
In this case \bb{note} that $v_1w_i,v_1s_j,v_1q_j\in E(G)$, $i\in\{1,2,3\}, j\in\{1,2\}$.  \bb{Further,} $v'v_2\in E(G)$ or $v'v_3\in E(G)$.  Assume without loss of generality that $v'v_2\in E(G)$.  \bb{Since $d^I(w_2)=d^I(w_3)=3$, note that $v'\cap N^I(w_2)\neq \emptyset$ and $v'\cap N^I(w_3)\neq \emptyset$.}   \bb{ Since $G$ is $2$-connected, observe that $v_3$ is adjacent to at least one more vertex $w'$ in $K$. If $w'=v'$, i.e., $v'v_3\in E(G)$, then $v'x_2\in E(G)$, otherwise \NC{w_2}{v'} or \NC{w_3}{v'} has an induced \K14.  Since $G$ has no short cycles, $d(v_2)>2$ or $d(v_3)>2$. That is, there exists more adjacency to $v_2$ or $v_3$.  If one of the end vertices of one among the paths $P_a, P_b, P_c$ is adjacent to either $v_2$ or $v_3$, then we obtain the desired path as follows.  If $w_1v_2\in E(G)$, then $P=(\overleftarrow{P_a},v_2,v,v_3,v',\overrightarrow{P_b},v_1,\overrightarrow{P_c})$ is a desired path. The other cases are similar. If $P_d\notin \{P_a, P_b, P_c\}$ has an end vertex  adjacent to either $v_2$ or $v_3$, say $v_2$, then we obtain $P=(\overrightarrow{P_d},v_2,v,v_3,v',\overrightarrow{P_a},v_1,\overrightarrow{P_b},\overrightarrow{P_c})$ as the desired path.         }  If $w'\neq v'$, then observe that $w'$ is an end vertex of some path $P_d$.  If $P_d\neq P_a$, $P_d\neq P_b$, and $P_d\neq P_c$, then $(\overrightarrow{P_d}w',v_3,v,v_2,v',\overrightarrow{P_a},v_1,\overrightarrow{P_b},\overrightarrow{P_c})$ is a desired path.  On the other hand if $P_d$ is one among $P_a,P_b,P_c$, say $P_d=P_b$, then $(\overrightarrow{P_b}w',v_3,v,v_2,v',\overrightarrow{P_a},v_1,\overrightarrow{P_c})$ is a desired path. 
\end{enumerate} 
\item  \emph{\bf {\em Case} 2:}  $w_2,w_3$ are adjacent to two different vertices in $N^I(v)$; i.e., without loss of generality, there exists $v_1,v_2\in N^I(v)$ such that  $v_1w_2,v_2w_3\in E(G)$.  Let $S=\{s_1,s_2,q_1,q_2\}$.
Since $G$ is $2$-connected, observe that $v_3$ is adjacent to at least one more vertex $w'$ in $K$.  We observe the following possibilities.
\begin{enumerate}[  ]
\item \emph{\bf {\em Case} 2.1:}  $w'\in P_a$.  Without loss of generality let $w_1v_3\in E(G)$.  Now note that $w_1$ is adjacent to a vertex in $N^I(w_3)$.  Thus we see the following sub cases.
\begin{enumerate}[  ]
\item \emph{\bf {\em Case} 2.1.1:} $w_1v_2\in E(G)$.\\
Observe that the vertices in $S$ are adjacent to any one of the following four vertex pairs; $\{v_2,x_1\}$, $\{v_2,x_2\}$, $\{v_2,v_1\}$, $\{v_3,x_2\}$.  
\\
\emph{\bf {\em Case} 2.1.1.1:} 
If there exists vertices in $S$ having adjacency to $v_2$.  Without loss of generality, $s_1v_2, q_1v_2\in E(G)$.  Clearly, $s_2,q_2$ are adjacent to one of $x_1,x_2,v_1$.  We shall see the following arguments with respect to $s_2$.\\
 If $s_2x_1\in E(G)$, then $P_1=(q_2,r_1,q_1,v_2,s_1,t_1,s_2,x_1,w_1,v_3,v,v_1,w_2,x_2,w_3,x_3,w_4)$ is a desired path in $G$.   
 If $s_2v_1\in E(G)$, then $P_2=(q_2,r_1,q_1,v_2,s_1,t_1,s_2,v_1,v,v_3,w_1,x_1,w_2,x_2,w_3,x_3,w_4)$ is a desired path in $G$.  
 Now  If $s_2x_2\in E(G)$, then we see two more possibilities.  If $s_2v_3\in E(G)$.  Note that in this case, $s_1,q_1$ are adjacent to $x_2$, otherwise \NC{q_1}{s_2} induces a \K14.  In this case $P_3=(q_2,r_1,q_1,x_2,s_1,t_1,s_2,v_3,v,v_1,w_2,$ $x_1,w_1,v_2,w_3,x_3,w_4)$ is a desired path in $G$. \\ 
 The final case remaining is when all the vertices in $S$ are adjacent to $v_2$; i.e., $s_1v_2,s_2v_2,q_1v_2,q_2v_2\in E(G)$.  Further, observe that if any of the vertex in $S$ are adjacent to $x_1$ or $v_1$, then we could obtain a similar path as that of $P_1$, $P_2$, respectively.  Therefore, we shall see the case when all the vertices of $S$ are adjacent to $v_2,x_2$.  Now, note that $d^I(v_3)=d^I(v_1)=d^I(x_1)=2$.  Since $G$ has no short cycles, there exists a vertex $w''\in K$ such that $w''$ is adjacent to some vertices in $\{v_3,v_1,x_1\}$.  If $w''=w_4$, then we obtain the following desired paths.\\
If $w_4v_3\in E(G)$, then $P=(q_2,r_1,q_1,v_2,s_1,t_1,s_2,x_2,w_3,x_3,w_4,v_3,v,v_1,w_2,x_1,w_1)$. \\
If $w_4v_1\in E(G)$, then $P=(q_2,r_1,q_1,v_2,s_1,t_1,s_2,x_2,w_3,x_3,w_4,v_1,w_2,x_1,w_1,v_3,v)$. \\
If $w_4x_1\in E(G)$, then $P=(q_2,r_1,q_1,v_2,s_1,t_1,s_2,x_2,w_3,x_3,w_4,x_1,w_1,v_3,v,v_1,w_2)$. \\
If $w''$ is an end vertex of a path $P_d$, then we obtain the desired path as follows. \\
If $w''v_3\in E(G)$, then $P=(q_2,r_1,q_1,v_2,s_1,t_1,s_2,x_2,w_3,x_3,w_4,\overrightarrow{P_d}w'',v_3,v,v_1,w_2,x_1,w_1)$. \\
If $w''v_1\in E(G)$, then $P=(q_2,r_1,q_1,v_2,s_1,t_1,s_2,x_2,w_3,x_3,w_4,\overrightarrow{P_d}w'',v_1,w_2,x_1,w_1,v_3,v)$. \\
If $w''x_1\in E(G)$, then $P=(q_2,r_1,q_1,v_2,s_1,t_1,s_2,x_2,w_3,x_3,w_4,\overrightarrow{P_d}w'',x_1,w_1,v_3,v,v_1,w_2)$. \\
The above cases hold true with respect to the vertex $q_2$; i.e., when $q_2$ is adjacent to one of $x_1,x_2,v_1$.  \\
\emph{\bf {\em Case} 2.1.1.2:}  If all the vertices in $S$ are non-adjacent to $v_2$, then every vertices in $S$ are adjacent to the vertices $v_3,x_2$.  In this case we obtain the same desired path $P_3$ as mentioned in the previous case.
\item \emph{\bf {\em Case} 2.1.2:} $w_1x_2\in E(G)$\\
Similar to Case 2.1.1. 
\item \emph{\bf {\em Case} 2.1.3:} $w_1x_3\in E(G)$\\
Observe that all the vertices in $S$ are adjacent to any one of the following three vertex pairs; $\{v_1,x_3\}$, $\{v_2,x_1\}$, $\{v_3,x_2\}$.
If all the vertices in $S$ are adjacent to $\{v_1,x_3\}$, then note that for every $u\in U=\{v_2,v_3,x_1,x_2\}$, $d^I(u)=2$.  Since $G$ has no short cycles, there exists a vertex $w''\in K$, such that for some $u\in U$, $w''u\in E(G)$.  If $w''=w_4$, then we obtain the following desired paths in $G$.\\
If $w_4v_3\in E(G)$, then $P=(q_2,r_1,q_1,v_1,s_1,t_1,s_2,x_3,w_4,v_3,w_1,x_1,w_2,x_2,w_3,v_2,v)$. \\
If $w_4v_2\in E(G)$, then $P=(q_2,r_1,q_1,v_1,s_1,t_1,s_2,x_3,w_4,v_2,v,v_3,w_1,x_1,w_2,x_2,w_3)$. \\
If $w_4x_1\in E(G)$, then $P=(q_2,r_1,q_1,v_1,s_1,t_1,s_2,x_3,w_4,x_1,w_2,x_2,w_3,v_2,v,v_3,w_1)$. \\
If $w_4x_2\in E(G)$, then $P=(q_2,r_1,q_1,v_1,s_1,t_1,s_2,x_3,w_4,x_2,w_3,v_2,v,v_3,w_1,x_1,w_2)$. \\
If $w''$ is an end vertex of a path $P_d$, then we obtain the desired path as follows. \\
If $w''v_3\in E(G)$, then  $P=(q_2,r_1,q_1,v_1,s_1,t_1,s_2,x_3,w_4,\overrightarrow{P_d}w'',v_3,w_1,x_1,w_2,x_2,w_3,v_2,v)$. \\
If $w''v_2\in E(G)$, then  $P=(q_2,r_1,q_1,v_1,s_1,t_1,s_2,x_3,w_4,\overrightarrow{P_d}w'',v_2,v,v_3,w_1,x_1,w_2,x_2,w_3)$. \\
If $w''x_1\in E(G)$, then  $P=(q_2,r_1,q_1,v_1,s_1,t_1,s_2,x_3,w_4,\overrightarrow{P_d}w'',x_1,w_2,x_2,w_3,v_2,v,v_3,w_1)$. \\
If $w''x_2\in E(G)$, then  $P=(q_2,r_1,q_1,v_1,s_1,t_1,s_2,x_3,w_4,\overrightarrow{P_d}w'',x_2,w_3,v_2,v,v_3,w_1,x_1,w_2)$. \\
If all the vertices in $S$ are adjacent to $\{v_2,x_1\}$, then $(q_2,r_1,q_1,v_2,s_1,t_1,s_2,x_1,w_1,v_3,v,v_1,w_2,$ $x_2,w_3,$ $x_3,w_4)$ is a desired path in $G$.
If all the vertices in $S$ are adjacent to $\{v_3,x_2\}$, then $(q_2,r_1,q_1,x_2,s_1,t_1,s_2,$ $v_3,w_1,x_1,w_2,v_1,v,v_2,w_3,x_3,w_4)$ is a desired path in $G$.
\end{enumerate} 
\item \emph{\bf {\em Case} 2.2:}  $w'\nin P_a$. 
In this case we shall assume that $w_1,w_4$ are adjacent to $v_1$ or $v_2$.  
Without loss of generality, let $w'\in P_b$; i.e., $w'=s_1$, $s_1v_3\in E(G)$.    
Observe that all the vertices in $S$ are adjacent to any one of the following six vertex pairs; $\{v_1,x_2\}$, $\{v_1,x_3\}$, $\{v_2,v_1\}$, $\{v_2,x_1\}$, $\{v_2,x_2\}$, $\{v_3,x_2\}$.  Clearly, $s_1v_3,s_1x_2\in E(G)$.  Note that all the vertices $q_1,q_2$ are adjacent to $x_2$, otherwise suppose $q_1$ is not adjacent to the vertices $x_2$, then \NC{s_1}{q_1} induces a \K14.  Similar argument holds for $q_2$.  Moreover, $q_1$ is adjacent to a vertex in $v_1,v_2,v_3$.  
If $w_1v_1,w_4v_1\in E(G)$, then $(s_2,t_1,s_1,v_3,v,v_2,w_3,x_3,w_4,v_1,w_1,x_1,w_2,x_2,q_1,r_1,q_2)$ is a desired path in $G$. 
If $w_1v_2,w_4v_2\in E(G)$, then $(s_2,t_1,s_1,v_3,v,v_1,w_2,x_1,w_1,v_2,w_4,x_3,w_3,x_2,q_1,r_1,q_2)$ is a desired path in $G$. 
If $w_1v_1,w_4v_2\in E(G)$, then we see the adjacency of $q_1$.\\
If $q_1v_1\in E(G)$, then $(s_2,t_1,s_1,v_3,v,v_2,w_4,x_3,w_3,x_2,q_2,r_1,q_1,v_1,w_2,x_1,w_1)$ is a desired path in $G$. \\
If $q_1v_2\in E(G)$, then $(s_2,t_1,s_1,v_3,v,v_1,w_1,x_1,w_2,x_2,q_2,r_1,q_1,v_2,w_3,x_3,w_4)$ is a desired path in $G$. \\
If $q_1v_3\in E(G)$, then $(s_2,t_1,s_1,v_3,q_1,r_1,q_2,x_2,w_3,x_3,w_4,v_2,v,v_1,w_1,x_1,w_2)$ is a desired path in $G$. \\
If $w_1v_2,w_4v_1\in E(G)$, then we see the adjacency of $q_1$.\\
If $q_1v_1\in E(G)$, then $(s_2,t_1,s_1,v_3,v,v_2,w_1,x_1,w_2,v_1,q_1,r_1,q_2,x_2,w_3,x_3,w_4)$ is a desired path in $G$. \\
If $q_1v_2\in E(G)$, then $(s_2,t_1,s_1,v_3,v,v_1,w_2,x_1,w_1,v_2,q_1,r_1,q_2,x_2,w_3,x_3,w_4)$ is a desired path in $G$. \\
If $q_1v_3\in E(G)$, then $(s_2,t_1,s_1,v_3,q_1,r_1,q_2,x_2,w_3,x_3,w_4,v_1,v,v_2,w_1,x_1,w_2)$ is a desired path in $G$. \\
If $w'\in P_c$ or $w'\nin S$, then similar argument could be made and the desired path is obtained.
\end{enumerate} 
\end{enumerate}
This completes the case analysis and a proof of the claim. $\hfill \qed$ 
\end{proof} 

\begin{cl}\label{no555}
 If $\mathbb{P}_j=\emptyset,j\ge7$, and $\mathbb{P}_5\neq\emptyset$, then $|\mathbb{P}_5|\le2$. 
\end{cl}
\begin{proof}
For a contradiction assume that there exists paths $P_a,P_b,P_c\in \mathbb{P}_5$ such that $P_a=(w_1,w_2,w_3;x_1,x_2)$, $P_b=(s_1,s_2,s_3;t_1,t_2)$, and $P_c=(q_1,q_2,q_3;r_1,r_2)$.  
From Claim \ref{xunivtwopaths}, there exists $v_1\in N^I(v)$ such that $v_1w_2,v_1s_2,v_1q_2\in E(G)$. 
Since the clique $K$ is maximal, there exists $w'\in K$ such that $w'v_1\notin E(G)$.  Therefore, $w'$ is an end vertex of some path in $\mathbb{P}_i,i\in\{1,3,5\}$.  We see the following cases.\\
\emph{\bf {\em Case} 1:} $w'$ is an end vertex of a $P_5$.\\
Without loss of generality assume $w'=w_1$, i.e., $w_1v_1\nin E(G)$.  Note that $w_1v_2\in E(G)$ or $w_1v_3\in E(G)$.  Observe that either $w_1t_1\in E(G)$ or $w_1t_2\in E(G)$, otherwise, \NC{s_2}{w_1} induces a \K14.  Further, either $w_1r_1\in E(G)$ or $w_1r_2\in E(G)$, otherwise, \NC{q_2}{w_1} induces a \K14. Clearly, \NV{w_1} induces a \K14, and thus $w'\nin \mathbb{P}_5$.  \\
\emph{\bf {\em Case} 2:} $w'\in \mathbb{P}_3\cup \mathbb{P}_1$.\\
Note that $w'v_2\in E(G)$ or $w'v_3\in E(G)$.  Observe that either $w'x_1\in E(G)$ or $w'x_2\in E(G)$, otherwise, \NC{w_2}{w'} induces a \K14. 
Similarly, either $w't_1\in E(G)$ or $w't_2\in E(G)$, otherwise, \NC{s_2}{w'} induces a \K14.  Further, either $w'r_1\in E(G)$ or $w'r_2\in E(G)$, otherwise, \NC{q_2}{w'} induces a \K14. It follows that, \NV{w'}  has an induced \K14.
This contradicts the assumption that there exists three such paths $P_a,P_b,P_c$.  This completes the cases analysis and a proof. $\hfill\qed$
\end{proof}
\begin{cl}\label{HC553}
 If there exist $P_a,P_b\in\mathbb{P}_5,P_c\in\mathbb{P}_3$ and $G$ has no short cycles, then $G$ has a Hamiltonian cycle.
\end{cl} \vspace{-10pt}
\begin{proof}
Let $P_a,P_b\in \mathbb{P}_5, P_c\in \mathbb{P}_3$ such that $P_a=(w_1,w_2,w_3;x_1,x_2)$, $P_b=(s_1,s_2,s_3;t_1,t_2)$, and $P_c=(q_1,q_2;r_1)$.    From Claim \ref{no555}, it follows that there are no more $5$-vertex paths in $\mathbb{C}$ other than $P_a,P_b$.
By Claim \ref{xunivtwopaths}, there exists $v_1\in N^I(v)$ such that $v_1w_2,v_1s_2\in E(G)$.  Now we claim that $v_1q_1,v_1q_2\in E(G)$.  Suppose if $v_1q_1\nin E(G)$, then either $q_1v_2$ or $q_1v_3\in E(G)$.  Further, $q_1x_1\in E(G)$ or $q_1x_2\in E(G)$, otherwise \NC{w_2}{q_1} has an induced \K14.  Either $q_1t_1\in E(G)$ or $q_1t_2\in E(G)$, otherwise \NC{s_2}{q_1} has an induced \K14.  It follows that \NV{q_1} induces a \K14, a contradiction to the assumption that $v_1q_1\nin E(G)$.  Similarly, it is easy to see that $v_1q_2\in E(G)$.  Since the clique is maximal, there exists a vertex $w'\in K$ such that $w'v_1\nin E(G)$.  We observe the following cases depending on the choice of $w'$.
\begin{enumerate}[  ]
\item  \emph{\bf {\em Case} 1:} $w'\in P_a\cup P_b$.\\
Without loss of generality let $v_1w_3\nin E(G)$.  From Claim A, either $w_3v_2\in E(G)$ or $w_3v_3\in E(G)$.  Further, let us assume without loss of generality $w_3v_2\in E(G)$.  Now, observe that $w_3t_1\in E(G)$ or $w_3t_2\in E(G)$, otherwise \NC{s_2}{w_3} induces a \K14. 
Without loss of generality, we shall assume that $w_3t_2\in E(G)$.  We see the following cases.
\begin{enumerate}[  ]
\item \emph{\bf {\em Case} 1.1:} $v_1w_1\in E(G)$.\\
Note that $w_1$ is adjacent to one of the vertices in $N^I(w_3)$, otherwise \NC{w_3}{w_1} induces a \K14.  Observe that $q_i, i\in\{1,2\}$ is adjacent to one of the vertices in $N^I(w_3)$.  Thus \bb{$d^I(w_i)=d^I(q_j)=3, i\in\{1,2,3\}, j\in\{1,2\}$}.  
Now, if there exists $w''\in K$, $w''v_3\in E(G)$ such that $w''\notin K'=\{w_1,w_2,w_3,s_1,s_2,s_3,$ $q_1,q_2\}$, then $w''$ is an end vertex of some path $P_d\in \mathbb{P}_3\cup \mathbb{P}_1$, then we obtain the following path as a desired path.  $\overrightarrow{P_d}w'',v_3,v,v_2,w_3\overleftarrow{P_a}w_1,v_1,\overrightarrow{P_c},\overrightarrow{P_b}$.  If $w''\in K'$, then either $w''=s_1$ or \bb{$w''=s_3$}.  If $w''=s_1$, then observe that $(s_3\overleftarrow{P_b}s_1,v_3,v,v_2,w_3\overleftarrow{P_a}w_1,v_1,\overrightarrow{P_c})$ is a desired path.
\item \emph{\bf {\em Case} 1.2:}  $v_1w_1\nin E(G)$.\\
Note that $d^I(w_2)=d^(w_3)=d^I(s_2)=3$, and note that $w_1$ is adjacent to one of $v_2,v_3$ and also adjacent to one of $t_1,t_2$.  Further, if $w_1v_3\in E(G)$, then $w_1t_2\in E(G)$.  That is, either $w_1v_3,w_1t_2\in E(G)$ or $w_1v_2,w_1t_2\in E(G)$  or $w_1v_2,w_1t_1\in E(G)$.  We detail the cases as follows.
\begin{enumerate}[  ]
\item \emph{\bf {\em Case} 1.2.1:} $w_1v_3,w_1t_2\in E(G)$.\\
Note that $d^I(v_3)=d^I(v_2)=d^I(x_1)=d^I(x_2)=2$ and $S=\{v_2,v_3,x_1,x_2\}\cup N(S)$ has a short cycle.  Since $G$ has no short cycles, it follows that there exists $w''\in K$ such that for some $x\in S$, $w''x\in E(G)$.  If $w''$ is an end vertex of a path $P_d\neq P_a\neq P_b\neq P_c$, then one of the following is a desired path $P$. \\
If $w''v_3\in E(G)$, then $P=(\overrightarrow{P_d}w'',v_3,w_1\overrightarrow{P_a}w_3,v_2,v,v_1,\overrightarrow{P_c},\overrightarrow{P_b})$\\
If $w''v_2\in E(G)$, then $P=(\overrightarrow{P_d}w'',v_2,w_3\overleftarrow{P_a}w_1,v_3,v,v_1,\overrightarrow{P_c},\overrightarrow{P_b})$\\
If $w''x_1\in E(G)$, then $P=(\overrightarrow{P_d}w'',x_1,w_1,v_3,v,v_2,w_3,x_2,w_2,v_1,\overrightarrow{P_c},\overrightarrow{P_b})$\\
If $w''x_2\in E(G)$, then $P=(\overrightarrow{P_d}w'',x_2,w_3,v_2,v,v_3,w_1,x_1,w_2,v_1,\overrightarrow{P_c},\overrightarrow{P_b})$\\
If $w''\in P_b$, then the path $P_d$ could be replaced by the path $P_b$ and remove the last occurrence of $P_b$ to get the desired path in all of the above cases.
\item \emph{\bf {\em Case} 1.2.2:} $w_1v_2,w_1t_2\in E(G)$ or $w_1v_2,w_1t_1\in E(G)$.\\
Since $G$ is $2$-connected, there exists a vertex $w''\in K$ such that $w''$ is an end vertex of a path $P_d\neq P_a\neq P_b\neq P_c$ such that $w''v_3\in E(G)$, then $(\overrightarrow{P_d}w'',v,v_2,w_1\overrightarrow{P_a}w_3,t_2,s_3,s_1,t_1,s_2,v_1,\overrightarrow{P_c})$ is a desired path.  Note that $q_1,q_2$ are adjacent to a vertex in $N^I(w_3)$. Therefore $w''\nin P_c$.  If $w''\in P_b$, then we obtain the following observations.  If $s_1v_3\in E(G)$ and $s_3v_3\in E(G)$, then $(\overrightarrow{P_c},v_1,s_2,t_1,s_1,v_3,s_3,t_2,w_3\overleftarrow{P_a}w_1,v_2,v)$ is a desired path.  If $s_1v_3\in E(G)$ and $s_3v_3\nin E(G)$, then $(\overrightarrow{P_c},v_1,s_2,t_1,s_1,v_3,v,v_2,w_1\overrightarrow{P_a}w_3,t_2,s_3)$.  Finally we shall see the case that  $s_1v_3\nin E(G)$ and $s_3v_3\in E(G)$.  In this case, note that either $s_1v_1\in E(G)$ or $s_1v_2\in E(G)$.  If $s_1v_1\in E(G)$, then $(\overrightarrow{P_c},v_1,s_1\overrightarrow{P_b}s_3,v_3,v,v_2,w_3\overleftarrow{P_a}w_1)$ is a desired path.  If $s_1v_2\in E(G)$, then $(w_1\overrightarrow{P_a}w_3,v_2,s_1\overrightarrow{P_b}s_3,v_3,v,\overrightarrow{P_c})$ is a desired path. 
\end{enumerate}
\end{enumerate}
\item \emph{\bf {\em Case} 2:} $w'\nin P_a\cup P_b$.\\
In this case we assume that $v_1w\in E(G)$, $w\in \{w_i,s_i,q_j\}, i\in \{1,2,3\}, j\in \{1,2\}$.  Clearly $w'v_2\in E(G)$ or $w'v_3\in E(G)$.  Without loss of generality assume that $w'v_2\in E(G)$.  We have already shown that $w'\nin P_c$.  Thus $w'$ is an end vertex of a path $P_e$.  Since $G$ is 2-connected, $v_3$ is adjacent to end vertex of a path $P_f$.  Note that $P_f\neq P_e$, otherwise $C=(v_3,\overrightarrow{P_e},v_2,v,v_3)$ is a short cycle in $G$.  Thus $P=(\overrightarrow{P_e},v_2,v,v_3,\overrightarrow{P_f},\overrightarrow{P_a},v_1,\overrightarrow{P_b},\overrightarrow{P_c})$ is a desired path.  
Note that if $P_f$ is some paths among $P_a,P_b,P_c$, say $P_f=P_a$, then the desired path is $(\overrightarrow{P_e},v_2,v,v_3,\overrightarrow{P_f},\overrightarrow{P_b},v_1,\overrightarrow{P_c})$.  Similarly we could easily obtain if $P_f=P_b$ and $P_f=P_c$.
\end{enumerate}
This completes the case analysis and a proof.  $\hfill \qed$
\end{proof} 

\noindent
\textbf{Observation:}
In the following claims, we shall produce a cycle $C$ containing all the vertices of $N^I(v)$.  Consider a path $P_m\in \mathbb{C}$ which is not a subpath of $C$.  It is easy to observe that $P_m$ is adjacent to at least one vertex in $\{v_1,v_2,v_3\}$, say $v_1$.  Further, $(\overrightarrow{P_m},v_1\overrightarrow{C}v_1^-)$ is a desired path in $G$, where the vertices $v_1^-,v_1$ occur consecutively in $\overrightarrow{C}$.
\begin{cl}\label{333}
Let $P_a,P_b,P_c\in\mathbb{P}_3$, and $v_1\in N^I(v)$ such that $v_1$ is adjacent to end vertex of at least two paths in $P_a,P_b,P_c$, $|\mathbb{P}_5|\le1$, and $\mathbb{P}_j=\emptyset,~j\ge7$.  If $G$ has no short cycles, then $G$ has a Hamiltonian cycle.
\end{cl}
\begin{proof}
Let $P_a=(w_1,w_2;x_1)$, $P_b=(s_1,s_2;t_1)$, $P_c=(q_1,q_2;r_1)$.  Without loss of generality, assume that $w_2v_1,s_1v_1\in E(G)$.
Since the clique is maximal, there exists $w'\in K$ such that $w'v_1\nin E(G)$.  We see the following cases. 
\begin{enumerate}[  ]
\item  \emph{\bf {\em Case} 1:}  $w'$ is one among $w_1,s_2$.\\
 Without loss of generality, let us assume that $v_1w_1\nin E(G)$.  Note that either $w_1v_2\in E(G)$ or $w_1v_3\in E(G)$.  Let us assume without loss of generality that $w_1v_2\in E(G)$.  Since $G$ is $2$-connected, $v_3$ is adjacent to at least one more vertex $w''\in K$.  Depending on the possibilities for $w''$ we see the following cases.
\begin{enumerate}[  ]
\item{ \emph{\bf {\em Case} 1.1:} $w''=w_1$.  Note that $q_1$ is adjacent to one of the vertices in $\{v_2,v_3,x_1\}$ otherwise \NC{w_1}{q_1} has an induced \K14.  
			If $q_1v_2\in E(G)$, then $P=(\overleftarrow{P_c}q_1,v_2,v,v_3,w_1,x_1,w_2,v_1,s_1,t_1,s_2)$ is a desired path.
			If $q_1v_3\in E(G)$, then $P=(\overleftarrow{P_c}q_1,v_3,v,v_2,w_1,x_1,w_2,v_1,s_1,t_1,s_2)$ is a desired path.
			
\bb{If $q_1x_1\in E(G)$, then note that $d(v_2)=d(v_3)=2$ and $N(v_2)=N(v_3)$.  Since $G$ has no short cycles, there exists a vertex $z'$ such that $z'v_2$ or $z'v_3$ is in $E(G)$. If $z'=w_2$, then we observe the following.  If $w_2v_2\in E(G)$, then $P=(\overleftarrow{P_c}q_1,x_1,w_1,v_3,v,v_2,w_2,v_1,s_1,t_1,s_2)$ is a desired path. If $w_2v_3\in E(G)$, then $P=(\overleftarrow{P_c}q_1,x_1,w_1,v_2,v,v_3,w_2,v_1,s_1,t_1,s_2)$ is a desired path. If $z'\in\{s_1,s_2\}$, then without loss of generality, we consider the case $z'=s_1$. If $s_1v_2\in E(G)$, then $P=(\overleftarrow{P_b}s_1,v_2,w_1,v_3,v,v_1,w_2,x_1,q_1\overrightarrow{P_c})$ is a desired path. If $s_1v_3\in E(G)$, then $P=(\overleftarrow{P_b}s_1,v_3,w_1,v_2,v,v_1,w_2,x_1,q_1\overrightarrow{P_c})$ is a desired path. If $z'\in\{q_1,q_2\}$ or $z'$ is an end vertex of a path $P_d\nin \{P_b,P_c\}$, then without loss of generality, we consider the case $z'=q_1$.  If $q_1v_2\in E(G)$, then $P=(\overleftarrow{P_c}q_1,v_2,v,v_3,w_1,x_1,w_2,v_1,s_1\overrightarrow{P_b})$ is a desired path.  If $q_1v_3\in E(G)$, then $P=(\overleftarrow{P_c}q_1,v_3,v,v_2,w_1,x_1,w_2,v_1,s_1\overrightarrow{P_b})$ is a desired path.  
}
}	
		\item \emph{\bf {\em Case} 1.2:} $w''=w_2$.
			Note that $d^I(v_3)=d^I(v_2)=d^I(x_1)=2$.  Since $G$ has no short cycles, at least one vertex of $v_3,v_2,x_1$ has at least one more adjacency in $K$.  If $w_1v_3\in E(G)$, then by \emph{Case 1.1}, there exists a desired path.  In all further possibilities, we produces a cycle $C$ containing all the vertices in $N^I(v)$, and path $P_c$ not in $C$.  Thus by previous observation, there exists a desired path in $G$.  \\
			If $s_2v_3\in E(G)$, then $C=(s_2,t_1,s_1,v_1,w_2,x_1,w_1,v_2,v,v_3,s_2)$.\\
			If $s_2v_2\in E(G)$, then $C=(s_2,t_1,s_1,v_1,v,v_3,w_2,x_1,w_1,v_2,s_2)$.\\			
			If $s_2x_1\in E(G)$, then $C=(s_2,t_1,s_1,v_1,w_2,v_3,v,v_2,w_1,x_1,s_2)$.\\
			If $s_1$ is adjacent to a vertex in $\{v_3,v_2,x_1\}$, then consider the adjacency of $s_2$.  That is,
			$s_2$ is adjacent to one of $v_1,v_2,v_3$.  Since we have seen the case when $s_2v_3\in E(G)$, and $s_2v_2\in E(G)$, the remaining case is when $s_2v_1\in E(G)$. \\
			If $s_2v_1,s_1v_3\in E(G)$, then  $C=(s_2,t_1,s_1,v_3,v,v_2,w_1,x_1,w_2,v_1,s_2)$.\\
			If $s_2v_1,s_1v_2\in E(G)$, then  $C=(s_2,t_1,s_1,v_2,w_1,x_1,w_2,v_3,v,v_1,s_2)$.\\
			If $s_2v_1,s_1x_1\in E(G)$, then  $C=(s_2,t_1,s_1,x_1,w_1,v_2,v,v_3,w_2,v_1,s_2)$.\\
			If there exists a vertex $w^*\nin \{w_1,s_1,s_2\}$ adjacent to a vertex in $\{v_3,v_2,x_1\}$, then note that $w^*$ is an end vertex of some path $P_d$.  We get desired path as follows.\\
			If $w^*x_1\in E(G)$, then $P=(\overleftarrow{P_d}w^*,x_1,w_1,v_2,v,v_3,w_2,v_1,s_1,t_1,s_2)$ is a desired path.\\
			If $w^*v_2\in E(G)$, then $P=(\overleftarrow{P_d}w^*,v_2,w_1,x_1,w_2,v_3,v,v_1,s_1,t_1,s_2)$ is a desired path.\\
			If $w^*v_3\in E(G)$, then $P=(\overleftarrow{P_d}w^*,v_3,v,v_2,w_1,x_1,w_2,v_1,s_1,t_1,s_2)$ is a desired path.
		\item \emph{\bf {\em Case} 1.3:} $w''=s_1$.
			Here $N^I(s_1)=\{v_1,t_1,v_3\}$.  Clearly, $q_1$ is adjacent to one of the vertices in $\{v_1,v_2,v_3\}$.\\
			If $q_1v_1\in E(G)$, then $P=(\overleftarrow{P_c}q_1,v_1,w_2,x_1,w_1,v_2,v,v_3,s_1,t_1,s_2)$ is a desired path.\\
			If $q_1v_2\in E(G)$, then $P=(\overleftarrow{P_c}q_1,v_2,w_1,x_1,w_2,v_1,v,v_3,s_1,t_1,s_2)$ is a desired path.\\
			If $q_1v_3\in E(G)$, then $P=(\overleftarrow{P_c}q_1,v_3,v,v_2,w_1,x_1,w_2,v_1,s_1,t_1,s_2)$ is a desired path.			
		\item \emph{\bf {\em Case} 1.4:} $w''=s_2$. \\
			In this case, note that the path $P_c$ is not in the cycle 
			$C=(s_2,t_1,s_1,v_1,w_2,x_1,w_1,v_2,v,v_3,s_2)$, and by previous Observation, we could easily get a desired path in $G$.
		\item \emph{\bf {\em Case} 1.5:} $w''\nin\{w_1,w_2,s_1,s_2\}$.
			Clearly $w''$ is an end vertex of some path $P_d$, and we obtain
			$P=(\overrightarrow{P_d}w'',v_3,v,v_2,w_1,x_1,w_2,v_1,s_1,t_1,s_2)$ is a desired path. 
\end{enumerate}
 \item  \emph{\bf {\em Case} 2:}  $w'$ is in $K\setminus \{w_1,s_2\}$.\\
Let $P_e=(y_1,y_2,y_3;z_1,z_2)$.  In this case note that for every $u\in W=\{w_1,w_2,s_1,s_2\}$, $v_1u\in E(G)$.  Since the clique is maximal, there exists $w^*\in K$ such that $v_1w^*\nin E(G)$, and $w^*$ is adjacent to one of $v_2,v_3$.  Without loss of generality, let $w^*v_2\in E(G)$.  Now $v_3$ is adjacent to at least one more vertex $s^*\in K$.  We see the following cases.
 \begin{enumerate}[  ]
\item \emph{\bf {\em Case} 2.1:} $w^*$ and $s^*$ are end vertices of two different paths, say $P_n,P_m$, respectively.\\
In this case, if $P_m\neq P_a$ and $P_m\neq P_b$, then $\overrightarrow{P_a},v_1,\overrightarrow{P_b},\overrightarrow{P_n}w^*,v_2,v,v_3,s^*\overleftarrow{P_m}$ is a desired path.
If $P_m$ is one among $P_a,P_b$ say $v_3w_1\in E(G)$, then $\overrightarrow{P_b},v_1,w_2,x_1,w_1,v_3,v,v_2,w^*\overleftarrow{P_n}$ is a desired path.
\item \emph{\bf {\em Case} 2.2:} $w^*$ and $s^*$ are end vertices of the same path, say $P_n$.
\begin{enumerate}[  ]
\item \emph{\bf {\em Case} 2.2.1:} $P_n \in \mathbb{P}_3$, say $P_n=P_c$, i.e., $q_1v_3,q_2v_2\in E(G)$. \\
Now it is easy to see that $d^I(v_2)=d^I(v_3)=d^I(r_1)=2$.  Since $G$ has no short cycles there exists $q^*\in K$ such that $q^*$ is adjacent to at least one of $v_2,v_3,r_1$.  If $q^*\in \{w_1,w_2,s_1,s_2\}$, then we obtain the following desired paths.
Without loss of generality let $q^*=w_1$.  \\
If $w_1r_1\in E(G)$, then $P_1=(\overrightarrow{P_b},v_1,w_2,x_1,w_1,r_1,q_2,v_2,v,v_3,q_1)$.\\
If $w_1v_2\in E(G)$, then $P_2=(\overrightarrow{P_b},v_1,w_2,x_1,w_1,v_2,v,v_3,q_1,r_1,q_2)$. \\
If $w_1v_3\in E(G)$, then $P_3=(\overrightarrow{P_b},v_1,w_2,x_1,w_1,v_3,v,v_2,q_2,r_1,q_1)$.
Note that if $q^*\in \{q_1,q_2\}$, say $q_1$, then $q_1v_2\in E(G)$.  Further, $d^I(q_1)=3$, and for every $u\in \{w_1,w_2,s_1,s_2\}$, $u$ is adjacent to one of $v_3,v_2,r_1$.  It follows that the above paths $P_1,P_2,P_3$ will be the desired paths.  
A symmetric argument holds if $q^*=q_2$.
Finally, if $q^*$ is some vertex other than $w_1,w_2,s_1,s_2,q_1,q_2$, then observe that $q^*$ is either an end vertex of some path $P_d$ or $q^*=y_2$, the middle vertex of path $P_e$.
 If $q^*$ is an end vertex of $P_d$, then we obtain the following desired paths.\\
 If $q^*r_1\in E(G)$, then $P_1=(\overrightarrow{P_b},v_1,w_2,x_1,w_1,\overrightarrow{P_d}q^*,r_1,q_2,v_2,v,v_3,q_1)$.\\
 If $q^*v_2\in E(G)$, then $P_2=(\overrightarrow{P_b},v_1,w_2,x_1,w_1,\overrightarrow{P_d}q^*,v_2,v,v_3,q_1,r_1,q_2)$. \\
 If $q^*v_3\in E(G)$, then $P_3=(\overrightarrow{P_b},v_1,w_2,x_1,w_1,\overrightarrow{P_d}q^*,v_3,v,v_2,q_2,r_1,q_1)$.
 Now, if $q^*=y_2$, then note that $d^I(y_2)=3$.  It follows that $w_1$  is adjacent to a vertex in $N^I(y_2)$.  Since we have already seen the case when $q^*=w_1$, we shall see the remaining cases as follows.\\
 If $w_1z_1,y_2r_1\in E(G)$, then $P=(y_1,z_1,w_1\overrightarrow{P_a},v_1,\overrightarrow{P_b},y_3,z_2,y_2,r_1,q_2,v_2,v,v_3,q_1)$.\\
 If $w_1z_1,y_2v_2\in E(G)$, then $P=(y_1,z_1,w_1\overrightarrow{P_a},v_1,\overrightarrow{P_b},y_3,z_2,y_2,v_2,v,v_3,q_1,r_1,q_2)$.\\
 If $w_1z_1,y_2v_3\in E(G)$, then $P=(y_1,z_1,w_1\overrightarrow{P_a},v_1,\overrightarrow{P_b},y_3,z_2,y_2,v_3,v,v_2,q_2,r_1,q_1)$.\\
 If $w_1z_2,y_2r_1\in E(G)$, then $P=(y_3,z_2,w_1\overrightarrow{P_a},v_1,\overrightarrow{P_b},y_1,z_1,y_2,r_1,q_2,v_2,v,v_3,q_1)$.\\
 If $w_1z_2,y_2v_2\in E(G)$, then $P=(y_3,z_2,w_1\overrightarrow{P_a},v_1,\overrightarrow{P_b},y_1,z_1,y_2,v_2,v,v_3,q_1,r_1,q_2)$.\\
 If $w_1z_2,y_2v_3\in E(G)$, then $P=(y_3,z_2,w_1\overrightarrow{P_a},v_1,\overrightarrow{P_b},y_1,z_1,y_2,v_3,v,v_2,q_2,r_1,q_1)$.
\item \emph{\bf {\em Case} 2.2.2:} $P_n \in \mathbb{P}_5$, say $P_n=P_e$, i.e., $y_1v_3,y_3v_2\in E(G)$.\\
Clearly, $d^I(v_2)=d^I(v_3)=d^I(z_1)=d^I(z_2)=2$.  Since $G$ has no short cycles there exists $q^*\in K$ such that $q^*$ is adjacent to at least one of $v_2,v_3,z_1,z_2$.  If $q^*\in \{w_1,w_2,s_1,s_2\}$, then we obtain the following desired paths.
Without loss of generality let $q^*=w_1$.  \\
If $w_1v_2\in E(G)$, then $P_1=(\overrightarrow{P_b},v_1,w_2,x_1,w_1,v_2,v,v_3,y_1\overrightarrow{P_e})$.\\
If $w_1v_3\in E(G)$, then $P_2=(\overrightarrow{P_b},v_1,w_2,x_1,w_1,v_3,v,v_2,y_3\overleftarrow{P_e})$.\\
If $w_1z_1\in E(G)$, then $P_3=(\overrightarrow{P_b},v_1,w_2,x_1,w_1,z_1\overrightarrow{P_e}y_3,v_2,v,v_3,y_1)$.\\
If $w_1z_2\in E(G)$, then $P_4=(\overrightarrow{P_b},v_1,w_2,x_1,w_1,z_2,y_3,v_2,v,v_3,y_1\overrightarrow{P_e}y_2)$.\\
If $q^*\in \{y_1,y_2,y_3\}$, say $y_1$ then note that $d^I(y_1)=3$.  Since \NC{y_1}{w_1} has no induced \K14, $w_1$ is adjacent to some vertices in $v_2,v_3,z_1,z_2$, and we obtain the desired paths $P_1,P_2,P_3,P_4$, same as above. 
If $q^*\nin \{w_1,w_2,s_1,s_2,y_1,y_2,y_3\}$, then we obtain the following desired paths.  Note that $q^*$ is an end vertex of some path $P_d\in \mathbb{P}_3$. \\
 If $q^*v_2\in E(G)$, then $P_1=(\overrightarrow{P_b},v_1,w_2,x_1,w_1,\overrightarrow{P_d}q^*,v_2,v,v_3,y_1\overrightarrow{P_e})$.\\
 If $q^*v_3\in E(G)$, then $P_2=(\overrightarrow{P_b},v_1,w_2,x_1,w_1,\overrightarrow{P_d}q^*,v_3,v,v_2,y_3\overleftarrow{P_e})$.\\
 If $q^*z_1\in E(G)$, then $P_3=(\overrightarrow{P_b},v_1,w_2,x_1,w_1,\overrightarrow{P_d}q^*,z_1\overrightarrow{P_e}y_3,v_2,v,v_3,y_1)$.\\
 If $q^*z_2\in E(G)$, then $P_4=(\overrightarrow{P_b},v_1,w_2,x_1,w_1,\overrightarrow{P_d}q^*,z_2,y_3,v_2,v,v_3,y_1\overrightarrow{P_e}y_2)$.
\end{enumerate} 
\item {\bf {\em Case} 2.3:} one of $w^*$ and $s^*$ is adjacent to $y_2$.\\
Without loss of generality let us assume $v_2y_2\in E(G)$.  Note that $d^I(y_2)=3$.  We now claim that for every $u\in \{w_1,w_2,s_1,s_2\}$, $uv_3\nin E(G)$.  Suppose there exists an adjacency for $v_3$ in $P_a\cup P_b$, say $v_3w_1\in E(G)$, then \NC{y_2}{w_1} induces a \K14.  Therefore, the vertex $s^*$ is such that $s^*\nin \{w_1,w_2,s_1,s_2\}$, $q^*v_3\in E(G)$.  
We have the following possibilities for $s^*$.
\begin{enumerate}[  ]
\item \emph{\bf {\em Case} 2.3.1:} $s^*\in \{y_1,y_3\}$.\\
If $v_3y_1\in E(G)$, then we obtain the desired path as follows.  Now $d^I(v_2)=d^I(v_3)=d^I(z_1)=2$.  Since $G$ has no short cycles, there exists a vertex $q^*\in K$ such that $q^*$ is adjacent to at least one of $v_2,v_3,z_1$.  
If $q^*\in \{w_1,w_2,s_1,s_2\}$, then without loss of generality, let us assume that $q^*=w_1$.  We already observed that $w_1v_3\nin E(G)$.  Thus if $w_1v_2\in E(G)$, then $P_1=(\overrightarrow{P_b},v_1,w_2,x_1,w_1,v_2,v,v_3,y_1\overrightarrow{P_e}y_3)$.\\
If $w_1z_1$, then $P_2=(\overrightarrow{P_b},v_1,w_2,x_1,w_1,z_1,y_1,v_3,v,v_2,y_2,z_2,y_3)$.\\
If $q^*=y_1$, then $y_1v_2\in E(G)$, $d^I(y_1)=3$, and $w_1$ is adjacent to at least a vertex in $N^I(y_1)$.  Note that $w_1v_3\nin E(G)$, and for the other two possibilities, we obtain the desired paths $P_1,P_2$ same as above. 
If $q^*=y_3$.  Note that all the vertices in $W=\{w_1,w_2,s_1,s_2\}$ are adjacent to a vertex in $N^I(y_2)$.  Since the case in which vertices in $W$ are adjacent to $v_2,z_1$ are already analysed, without loss of generality we shall assume that for every $u\in \{w_1,w_2,s_1,s_2\}$, $uz_2\in E(G)$.  \\
If $y_3v_3\in E(G)$, then $P=(\overrightarrow{P_b},v_1,w_2,x_1,w_1,z_2,y_3,v_3,v,v_2,y_2,z_1,y_1)$.\\
If $y_3v_2\in E(G)$, then $P=(\overrightarrow{P_b},v_1,w_2,x_1,w_1,z_2,y_3,v_2,v,v_3,y_1,z_1,y_2)$.\\
If $y_3z_1\in E(G)$, then $P=(\overrightarrow{P_b},v_1,w_2,x_1,w_1,z_2,y_3,z_1,y_2,v_2,v,v_3,y_1)$.\\
Now we shall consider $q^*\nin \{w_1,w_2,s_1,s_2,y_1,y_3\}$.  Here also, similar to the previous argument we shall assume that for every $u\in \{w_1,w_2,s_1,s_2\}$, $uz_2\in E(G)$.  Note that $q^*$ is an end vertex of some path $P_d\in \mathbb{P}_3$. \\
If $q^*v_3\in E(G)$, then $P=(\overrightarrow{P_b},v_1,w_2,x_1,w_1,z_2,y_3,\overrightarrow{P_d}q^*,v_3,v,v_2,y_2,z_1,y_1)$.\\
If $q^*v_2\in E(G)$, then $P=(\overrightarrow{P_b},v_1,w_2,x_1,w_1,z_2,y_3,\overrightarrow{P_d}q^*,v_2,v,v_3,y_1,z_1,y_2)$.\\
If $q^*z_1\in E(G)$, then $P=(\overrightarrow{P_b},v_1,w_2,x_1,w_1,z_2,y_3,\overrightarrow{P_d}q^*,z_1,y_2,v_2,v,v_3,y_1)$.\\
If $v_3y_3\in E(G)$, then the case is symmetric. 
\item \emph{\bf {\em Case} 2.3.2:} $s^*\nin \{y_1,y_3\}$.\\
Note that $s^*$ is an end vertex of a path $P_f\in \mathbb{P}_3$.
Note that $w_1$ is adjacent to at least a vertex in $N^I(y_2)$.  We obtain the desired paths as follows.\\
If $w_1v_2\in E(G)$, then $P=(\overrightarrow{P_b},v_1,w_2,x_1,w_1,v_2,v,v_3,q^*\overrightarrow{P_f},\overrightarrow{P_e})$.\\
If $w_1z_1\in E(G)$, then $P=(\overrightarrow{P_b},v_1,w_2,x_1,w_1,z_1,y_1,y_3,z_2,y_2,v_2,v,v_3,q^*\overrightarrow{P_f})$.\\
If $w_1z_2\in E(G)$, then $P=(\overrightarrow{P_b},v_1,w_2,x_1,w_1,z_2,y_3,y_1,z_1,y_2,v_2,v,v_3,q^*\overrightarrow{P_f})$.\\
\end{enumerate}


 \end{enumerate}
\end{enumerate} 
\vspace{-16pt}
This completes the case analysis and the proof. $\hfill\qed$
\end{proof}

\begin{cl}\label{HC5333}
 If there exist $P_a\in\mathbb{P}_5,P_b,P_c,P_d\in\mathbb{P}_3$ and $G$ has no short cycles, then $G$ has a Hamiltonian cycle.
\end{cl}
\begin{proof}
Note that there are four paths and thus there exists a vertex in $N^I(v)$, adjacent to at least two different paths.  Without loss of generality, let us assume that the end vertices of two different paths are adjacent to $v_1\in N^I(v)$.  If two paths, say $P_b,P_c\in \mathbb{P}_3$ are adjacent to $v_1$, then by Claim \ref{333}, $G$ has a Hamiltonian cycle.  On the other hand, we shall assume that no such two paths exists in $\mathbb{P}_3$.  Therefore, we could assume that the end vertices of $P_b,P_c,P_d$ are adjacent to $v_1,v_2,v_3$, respectively.  Now end vertices of $P_a$ is adjacent to a vertex in $N^I(v)$, say $v_1$.  Observe that $P=(\overrightarrow{P_a},v_1,\overrightarrow{P_b},\overrightarrow{P_c},v_2,v,v_3,\overrightarrow{P_d})$ is a desired path in $G$.  This completes the proof.  $\hfill\qed$

\end{proof} 

\begin{cl}\label{HC33333}
 If $\mathbb{P}_j=\emptyset,j\ge5$ and there exists $P_a,P_b,P_c,P_d,P_e\in\mathbb{P}_3$ and $G$ has no short cycles, then $G$ has a Hamiltonian cycle.
\end{cl}
\begin{proof}
Note that there exists five paths and at least two those paths are adjacent to one of $v_1,v_2,v_3$.  Observe that the premise of Claim \ref{333} is satisfied, and therefore, $G$ has a Hamiltonian cycle.
$\hfill\qed$
\end{proof} 
\begin{theorem}
Let $G$ be a $2$-connected, $K_{1,4}$-free split graph with $|K|\ge|I|\ge8$. 
$G$ has a Hamiltonian cycle if and only if there are no induced short cycles in $G$.  
Further, finding such a cycle is polynomial-time solvable.
\end{theorem}
\begin{proof}
Necessity is trivial.  Sufficiency follows from the previous claims.  \hfill\qed
\end{proof}



\section{Hardness Result}
Akiyama et al. \cite{akiyama} proved the NP-completeness of Hamiltonian cycle in planar bipartite graphs with maximum degree $3$.  It is easy to see that if the bipartite graph is having different sized partitions for the vertex set, then it is clearly a NO instance for the Hamiltonian cycle problem.  It follows that the Hamiltonian cycle problem in planar bipartite graphs with maximum degree $3$, and equal sized vertex partitions is NP-hard.  Here we give a reduction from Hamiltonian cycle problem in planar bipartite graphs with maximum degree $3$, and equal sized vertex partitions to Hamiltonian cycle problem in $K_{1,5}$-free split graph. 
\begin{theorem} \label{k15hamil}
HCYCLE in $K_{1,5}$-free split graph is NP-complete.
\end{theorem}
\begin{proof} For NP-hardness result, we present a deterministic polynomial-time reduction that reduces an instance of planar bipartite graph with maximum degree $3$, and equal sized vertex partitions to a corresponding split graph instance. 
Consider such a graph $G$ with maximum degree $3$, and let $A,B$ be the partitions of $V(G)$.  We construct a split graph $H$ from $G$ as follows. \\\\
$V(H)=V(G), E(H)=E(G)\cup E'$ where 
 $E'=\{uv~:~u,v\in A\}$\\\\
Clearly, the reduction is a polynomial-time reduction and $H$ is a split graph with a maximal clique $A$.
We now show that there exists a Hamiltonian cycle in $G$ if and only if there exists a Hamiltonian cycle in $H$.  
 \emph{Necessity:} If there exists a Hamiltonian cycle $C$ in $G$, then $C$ is a Hamiltonian cycle in $H$, since $G$ is a strict subgraph of $H$.\\
\emph{Sufficiency:}  
We claim that any Hamiltonian cycle $C$ in $H$ is a Hamiltonian cycle in $G$.  If not, there exists at least one edge $uv\in E(C)$ where $u,v\in A$.  It follows that at least one vertex in $B$ is not in $C$, which contradicts the Hamiltonicity of $H$.  Hence the sufficiency follows.
We now show that the constructed graph $H$ is $K_{1,5}$-free split graph.  Suppose there exists a $K_{1,5}$ in $H$ induced on vertices $\{u,v,w,x,y,z\}$, centered at $v$.  At most two vertices (say $u,v$) of $K_{1,5}$ belongs to the clique $A$.  Therefore, $w,x,y,z\in B$ and this implies $d_{H}^{I}(v)\ge4$.  It follows that $d_{G}(v)\ge4$, which is a contradiction to the maximum degree of the bipartite graph $G$.  Since a given instance of Hamiltonian cycle problem in $K_{1,5}$-free split graph can be verified in deterministic polynomial time, the problem is in class NP.  It follows that the Hamiltonian cycle problem in $K_{1,5}$-free split graphs is NP-complete and the theorem follows. \hfill \qed
\end{proof}

\begin{corollary}
HCYCLE in $K_{1,r}$-free split graph, $r \geq 5$ is NP-complete.
\end{corollary}
\begin{proof}
Follows from Theorem 3. \hfill \qed
\end{proof}

\bibliographystyle{splncs2}
\bibliography{hcycle_split}
\end{document}